%% file: main.tex
\documentclass[anonymous,11pt]{article}

\usepackage{multicol}
\usepackage{amsmath, amsthm, amsfonts, amssymb, amsbsy, anyfontsize, bbm, bm, color, cases, dsfont, graphicx, enumerate, caption, mathrsfs, parskip, subfig, setspace, tikz, mathtools, newpxtext, newpxmath}

\usepackage[margin=1in]{geometry}
\usepackage{hyperref}
\usepackage{orcidlink}
\usepackage[capitalise]{cleveref}

\usepackage{environ} %

\input{commands}

\allowdisplaybreaks

\usepackage{lipsum}
\usepackage{mathtools}

\usepackage{tikz-cd}

\DeclarePairedDelimiter\ceil{\lceil}{\rceil}
\DeclarePairedDelimiter\floor{\lfloor}{\rfloor}

\newcommand{\numberthis}[1]{\addtocounter{equation}{1}\tag{\theequation}\label{#1}}

\author{
Jacob Calvert \orcidlink{0000-0001-9173-0946}\\ 
\href{mailto:calvert@gatech.edu}{\texttt{calvert@gatech.edu}}
\\
{Institute for Data Engineering and Science, Georgia Institute of Technology, USA}
\and
Shunhao Oh\footnote{Corresponding author.} ~\orcidlink{0009-0002-1328-0040} \\ 
\href{mailto:ohoh@gatech.edu}{\texttt{ohoh@gatech.edu}}\\
{School of Computer Science, Georgia Institute of Technology, USA}
\and
Dana Randall \orcidlink{0000-0002-1152-2627}\\ 
\href{mailto:randall@cc.gatech.edu}{\texttt{randall@cc.gatech.edu}}
\\
{School of Computer Science, Georgia Institute of Technology, USA}
}

\title{The Geometry of Fixed-Magnetization Spin Systems\\ at Low Temperature}
\date{}

\begin{document}
\maketitle
\thispagestyle{empty}

\begin{abstract}
Spin systems are fundamental models of statistical physics that provide insight into collective behavior across scientific domains. Their interest to computer science stems in part from the deep connection between the phase transitions they exhibit and the computational complexity of sampling from the probability distributions they describe. Our focus is on the geometry of spin configurations, motivated by applications to programmable matter and computational biology. Rigorous results in this vein are scarce because the natural setting of these applications is the low-temperature, fixed-magnetization regime. The latter means that the numbers of each type of spin are held constant, a global constraint which makes analysis notoriously difficult. Recent progress in this regime is largely limited to spin systems under which magnetization concentrates, which enables the analysis to be reduced to that of the simpler, variable-magnetization case. More complicated models, like those that arise in applications, do not share this property. 

We study the geometry of spin configurations on the triangular lattice under the Generalized Potts Model (GPM). The GPM generalizes many fundamental models of statistical physics, including the Ising, Potts, clock, and Blume--Capel models. It moreover specializes to models used to program active matter to solve tasks like compression and separation, and it is closely related to the Cellular Potts Model, which is widely used in computational models of biological processes. Our main result shows that, under the fixed-magnetization GPM at low temperature, spins of different types are typically partitioned into regions of mostly one type, separated by boundaries that have nearly minimal perimeter. The proof uses techniques from Pirogov--Sinai theory to extend a classic Peierls argument for the fixed-magnetization Ising model, and introduces a new approach for comparing the partition functions of fixed- and variable-magnetization models. The new technique identifies a special class of spin configurations that contribute comparably to the two partition functions, which then serves as a bridge between the fixed- and variable-magnetization settings.
\end{abstract}

\newpage
\setcounter{page}{1}
\input{main_section}

\input{proof_section}

\bibliographystyle{alpha}
\bibliography{references.bib}

\input{proof_appendix}
\end{document}

%% file: commands.tex
\newtheorem{theorem}{Theorem}
\newtheorem{lemma}{Lemma}[section]

\newtheorem{corollary}[lemma]{Corollary}
\newtheorem{definition}[lemma]{Definition}

\newcommand{\redcomment}[1]{\iftrue\color{red} #1 \color{black}\fi}
\newcommand{\bluecomment}[1]{\iftrue\color{blue} #1 \color{black}\fi}

\newcommand{\R}{\mathbb{R}}

\newcommand{\E}{\mathbf{E}}

\newcommand{\eps}{\epsilon}

\newcommand{\h}{\hvec}

\newcommand{\ds}{\displaystyle}

\newcommand{\deriv}{\partial}

\newcommand{\derivset}{\{0,1,2\}}
\newcommand{\derivki}[2]{\frac{\deriv^{#1}}{\deriv h_{#2}^{#1}}}

\newcommand{\fulllattice}{{\Lambda^\infty}}
\newcommand{\regionL}{{\Lambda^L}}
\newcommand{\Lbox}{\Lambda^L}

\newcommand{\hamil}{H}
\newcommand{\Hcost}[1]{\mathcal{H}(#1)}
\newcommand{\Hcostbar}[1]{\mathcal{H}(#1)}%

\newcommand{\contourcost}[1]{\|#1\|_{\beta,\hvec}}

\newcommand{\conbad}{\mathcal{S}}
\newcommand{\conbadbar}{\overline{\mathcal{S}}}
\newcommand{\congood}{\mathcal{O}_{\beta}}
\newcommand{\congoodbar}{\overline{\mathcal{O}_{\beta}}}
\newcommand{\reggood}[2]{R^{#1}_{#2}} %

\newcommand{\nvecp}{\mathbf{n}'}
\newcommand{\nvecbeta}{\mathbf{n}^{\beta}}
\newcommand{\npj}[1]{n'_{#1}} %
\newcommand{\nbetaj}[1]{n_{#1}^{\beta}} %
\newcommand{\nfixedvec}{{\mathbf{n}}}
\newcommand{\nj}[1]{n_{#1}} %

\newcommand{\hvec}{{\mathbf{h}}}
\newcommand{\hvecopt}{{\mathbf{h}^*_{\beta}}}
\newcommand{\hoptj}[1]{h^*_{\beta, #1}} %
\newcommand{\mvec}{\mathbf{m}}
\newcommand{\rhovec}{{\boldsymbol{\rho}}}
\newcommand{\rhovecp}{{\boldsymbol{\rho}'}}
\newcommand{\rhovecbeta}{\boldsymbol{\rho}^{\langle\beta\rangle}}
\newcommand{\rhop}[1]{\rho'_{#1}} %
\newcommand{\rhobeta}[1]{\rho_{#1}^{\langle\beta\rangle}} %

\newcommand{\rhomin}{\rho_{\mathrm{min}}}
\newcommand{\rhopmin}{\rho'_{\mathrm{min}}}
\newcommand{\rhobetamin}{\rho^{\langle\beta\rangle}_{\mathrm{min}}}

\newcommand{\nowrapsymbol}{\times}

\newcommand{\ZFbase}{Z}
\newcommand{\ZFbd}[1]{\ZFbase_{#1}} %
\newcommand{\ZFtr}[1]{\widehat{\ZFbase}_{#1}} %
\newcommand{\ZFf}[1]{\widetilde{\ZFbase}_{*,#1}} %

\newcommand{\ZPbase}{Z} %
\newcommand{\ZPbd}[2]{\ZPbase^{#1}_{#2}} %
\newcommand{\ZPtr}[2]{\widehat{\ZPbase}^{#1}_{#2}} %
\newcommand{\ZPbdf}[3]{\widetilde{\ZPbase}^{#1}_{#2,#3}} %
\newcommand{\ZPf}[2]{\widetilde{\ZPbase}^{#1}_{#2}} %

\newcommand{\ZPcf}[3]{\widetilde{\ZPbase}^{#1}_{#2,#3}} %
\newcommand{\ZPc}[2]{{\ZPbase}^{#1}_{#2}} %

\newcommand{\ZPSbase}{Z} %
\newcommand{\ZPSbd}[2]{\ZPSbase^{\nowrapsymbol,#1}_{#2}} %
\newcommand{\ZPStr}[2]{\widehat{\ZPSbase}^{\nowrapsymbol,#1}_{#2}} %
\newcommand{\ZPSbdf}[3]{\widetilde{\ZPSbase}^{\nowrapsymbol,#1}_{#2,#3}} %

\newcommand{\CFGF}{\Omega}
\newcommand{\CFGFint}[2]{\CFGF_{#1}(#2)} %
\newcommand{\CFGPint}[3]{\CFGF^{#1}_{#2}(#3)} %
\newcommand{\CFGPintf}[4]{\CFGF^{#1}_{#2,#3}(#4)} %
\newcommand{\CFGPc}[2]{\CFGF^{#1}_{#2}} %
\newcommand{\CFGPcf}[3]{\widetilde{\CFGF}^{#1}_{#2,#3}} %
\newcommand{\CFGPf}[2]{\widetilde{\CFGF}^{#1}_{#2}} %

\newcommand{\CFGPSint}[3]{\CFGF^{\nowrapsymbol,#1}_{#2}(#3)} %

\newcommand{\pressuretr}[1]{\widehat{\psi}_{#1}} %

\newcommand{\partialext}{\partial_{\mathrm{ext}}}
\newcommand{\partialint}{\partial_{\mathrm{int}}}

\newcommand{\magrvi}[1]{M_{#1,j}}
\newcommand{\magrv}{M_{\cdot,j}}

\newcommand{\contour}{\gamma}
\newcommand{\contourbar}{\overline{\contour}}
\newcommand{\contourbarp}{\overline{\contour'}}
\newcommand{\contourset}{\Gamma}
\newcommand{\contoursetbar}{\overline{\contourset}}
\newcommand{\contoursetext}{\Gamma_{\mathrm{ext}}}
\newcommand{\cluster}{X}
\newcommand{\clusterbar}{\overline{\cluster}}

\newcommand{\confinal}{\Gamma}
\newcommand{\confinalbar}{\overline{\confinal}}

\newcommand{\contourcfgs}{\mathcal{P}}
\newcommand{\contourcfgsperbdv}[3]{\mathcal{P}^{\nowrapsymbol, #1}_{#2}(#3)} %

\newcommand{\allcontours}{\mathcal{C}}
\newcommand{\allcontoursbd}[1]{\mathcal{C}_{#1}} %
\newcommand{\allcontoursperbdv}[3]{\mathcal{C}^{\nowrapsymbol, #1}_{#2}(#3)} %

\newcommand{\clusterset}{\mathcal{X}}
\newcommand{\clustersetbd}[1]{\mathcal{X}_{#1}} %
\newcommand{\clustersetbdv}[2]{\mathcal{X}_{#1}(#2)} %
\newcommand{\clustersetperbdv}[3]{\mathcal{X}^{\nowrapsymbol, #1}_{#2}(#3)} %

\newcommand{\wutr}[1]{w_{#1}} %
\newcommand{\wtr}[1]{\widehat{w}_{#1}} %
\newcommand{\wtrb}[1]{\overline{w}_{#1}} %

\newcommand{\clw}{\Psi} %
\newcommand{\clwbd}[1]{\Psi_{#1}} %
\newcommand{\clwtr}[1]{\widehat{\Psi}_{#1}} %
\newcommand{\clwtrb}[1]{\overline{\Psi}_{#1}} %
\newcommand{\dclwtr}[1]{\derivki{k}{i} \widehat{\Psi}_{#1}} %

\newcommand{\intr}[1]{\mathrm{Int}_{#1}}

\newcommand{\extr}{\mathrm{Ext}}
\newcommand{\taub}{\tau_\beta}
\newcommand{\ursell}{\varphi}
\newcommand{\minlength}{\ell_0}
\newcommand{\contourconst}{\Delta}

\newcommand{\costmatrixmax}{A_{\mathrm{max}}}
\newcommand{\costmatrixmin}{A_{\mathrm{min}}}
\newcommand{\hetcost}{\costmatrixmin}

\newcommand{\numcolors}{q}
\newcommand{\colorset}{[\numcolors]}

\newcommand{\infconfig}{\sigma}
\newcommand{\infassg}{\omega}
\newcommand{\contourassg}[1]{\infconfig_{#1}}
\newcommand{\config}{\sigma}

\newcommand{\compl}{c}
\newcommand{\meanball}{\mathcal{M}_t}

\newcommand{\betamin}{\beta_1} %
\newcommand{\betaminA}{\beta_2}
\newcommand{\betaminB}{\beta_3}
\newcommand{\betaminF}{\beta_0} %
\newcommand{\alphaA}{{\alpha'}}
\newcommand{\alphaB}{{\alpha''}}
\newcommand{\Lmin}{L_1}
\newcommand{\LminF}{L_0}

\newcommand{\numparticles}{L^2}

\newcommand{\optconfig}{\sigma^*}
\newcommand{\optconfign}[1]{\sigma^{L}_{#1}} %

\newcommand{\dtbnh}{\pi_{\beta, \hvec}}
\newcommand{\dtbnf}{\widetilde{\pi}_{\beta, \rhovec}} %

\newcommand{\hoptbound}{3Qe^{-\taub \minlength/4}}
\newcommand{\smallchangesbound}{7\beta\costmatrixmin \sqrt{\numcolors}}

\newcommand{\poly}{\mathrm{poly}}

\newcommand{\bdconstupp}{b^\top}
\newcommand{\bdconstlow}{b^\bot}
\newcommand{\Hconstupp}{b_H^\top}
\newcommand{\Hconstlow}{b_H^\bot}

\newcommand{\hoptopenset}{\mathcal{W}_{\beta}}

\NewEnviron{colorred}{\color{red} \BODY \color{black}}{}

%% file: main_section.tex
\section{Introduction}

Spin systems from statistical physics are powerful tools for understanding the thermodynamics of interacting particle systems. %
The canonical example is the Ising model on a graph $G = (V,E)$, which concerns a probability distribution on configurations $\Omega = \{-1,1\}^V$ of {\em spins} on vertices. The probability of a configuration $\sigma \in \Omega$ is determined by its {\it Hamiltonian}, or energy, defined as
\[
H(\sigma)=-\sum_{(u,v) \in E} \sigma_u \sigma_v.
\]
In particular, $H(\sigma)$ is smaller when vertices with the same spin cluster together more. The {\em Gibbs distribution} assigns the configuration $\sigma$ a probability
\[
\pi(\sigma)=\frac{e^{-\beta H(\sigma)}}{Z},
\]
in terms of inverse temperature $\beta$ and normalization constant $Z=\sum_{\tau \in \Omega} e^{-\beta H(\tau)}$. The {\em fixed-magnetization} Ising model has the same Hamiltonian, but the corresponding Gibbs distribution is restricted to configurations with a given number of spins of each type.

The qualitatively different behavior of the Gibbs distribution at high and low temperature provides a classical explanation of physical phenomena, like the spontaneous magnetization of ferromagnets. More recently, however, the Ising model---and many related spin systems---have become important in computer science, due in part to the deep connection between the phase transitions they exhibit and algorithms for sampling from the distributions they describe \cite{Moore2011}. These connections highlight the importance of understanding the geometry of configurations that are typical under the Gibbs distribution, especially in the fixed-magnetization case.

For example, a variety of spin systems are important to the field of {\it programmable active matter}, which seeks to design simple computational elements, or ``particles,'' that self-organize to collectively solve problems \cite{Toffoli1991}. %
Instead of carefully choreographing the behavior of each particle, an emerging approach to programmable matter encodes desirable collective behavior in the Gibbs distribution of a spin system \cite{Cannon2016,Cannon2019,Li2021}. The corresponding Glauber dynamics, which is normally used to sample from the Gibbs distribution, then constitutes a stochastic distributed algorithm, or ``program,'' for the behavior in question. Critically, since the number of particles is fixed, the number of each type of spin is too. As a proof-of-concept, \cite{Li2021} physically embodied the Glauber dynamics of a fixed-magnetization Ising model in simple robots. The resulting swarm could tunably aggregate and collectively transport objects or disperse, its functionality reflecting the geometry of typical configurations under the Gibbs distribution. %

Spin systems are also widely used in computational models of biological processes, like cell migration, tissue morphogenesis, and tumor formation \cite{Shaebani2020,Alert2020}. Foremost among these models is the {\it Cellular Potts Model} (CPM) \cite{Graner1992,Glazier1993}, in which the spins model various types of cells and the surrounding medium, and the Hamiltonian models the forces that cells exert on each other and their environment \cite{Rens2019}. To enforce constraints on cell volume, the Hamiltonian further includes a term that penalizes deviations of the number of spins of each type from desired counts. %
The dynamics, which usually proceeds according to the Metropolis algorithm for sampling the corresponding Gibbs distribution, is not necessarily intended to accurately model the biological system's evolution \cite{Nemati2024}. Instead, it is the geometry of typical configurations that provides insight into ``innumerable'' biological phenomena \cite{Scianna2012}. %

These applications highlight the value of understanding the geometry of typical configurations under the Gibbs distribution of spin systems with fixed magnetization. To this end, we analyze a {\em Generalized Potts Model} (GPM), which addresses both classes of applications and generalizes many fundamental models of equilibrium statistical physics. %
\subsection{Generalized Spin Systems}
We start by defining the GPM. For convenience, we will refer to types of spins as colors. Let $G$ be a finite graph, let $q$ be a number of colors, and let $A \in \R^{q\times q}$ be a symmetric matrix that has zeroes along its diagonal and positive entries otherwise. We call $A$ a {\em cost matrix} because entry $A(i,j)$ is the energy cost of an edge with colors $i$ and $j$ at its endpoints. The Hamiltonian of the GPM is
\begin{equation}\label{eq:hamiltonian}
H_{G} (\sigma) := \sum_{(u,v) \in E(G)} A(\sigma_u,\sigma_v), \quad \sigma \in \Omega_G := [q]^{V(G)}.
\end{equation}
The corresponding family of %
Gibbs distributions on $\Omega_G$, indexed by inverse temperature $\beta \geq 0$, is 
\begin{equation}\label{eq:pi}
\pi_{G,\beta} (\sigma) := \frac{e^{-\beta H_G (\sigma)}}{Z_{G}(\beta)},
\end{equation}
where $Z_{G}(\beta) := \sum_{\sigma \in \Omega_G} e^{-\beta H_G (\sigma)}$ is the normalization constant known as the partition function.

The GPM specializes to many well-known models through different choices of the cost matrix. For example, the $q$-state ferromagnetic Potts model corresponds to the choice $A(i,j) = \mathbbm{1}_{i\neq j}$, which further specializes to the Ising model when $q=2$. Other special cases of the GPM include the Blume--Capel model, which corresponds to $q = 3$ and $A(i,j) = (i-j)^2$, and the $q$-state clock model, which corresponds to the cost matrix with entries $A(i,j) = 1 - \cos (2 \pi (i-j)/q)$.

These spin systems exhibit phase transitions as temperature varies, which can impact the efficiency of algorithms for sampling from their Gibbs distributions, or for approximating their partition functions. For example, the Glauber dynamics can provide an efficient means of approximately sampling from some of these Gibbs distributions, but only above a critical temperature, where the dynamics mix rapidly \cite{Sinclair1989,Borgs1999b}. However, this does not preclude the existence of efficient sampling algorithms in the low-temperature regime altogether, as a flurry of recent work demonstrates \cite{Helmuth2019,Helmuth2020,Jenssen2020,Borgs2020,Chen2021}. To analyze their sampling algorithms, these works rely on advanced techniques from the physics literature, known as Pirogov--Sinai theory \cite{Pirogov1975,Pirogov1976,Friedli2017}.

Despite advances in sampling from Gibbs distributions in the low-temperature regime, nearly all such results concern {\it variable-magnetization} models like \cref{eq:pi}, which allow color counts to vary.\footnote{Variable-magnetization models in our setting are instances of the {\em grand canonical ensemble} of statistical physics, while the fixed-magnetization models are instances of the {\em canonical ensemble}.} %
In contrast, fixed-magnetization models are notoriously difficult to analyze, especially in the low-temperature regime, because their partition functions generally lack a representation known as a {\em contour polymer model}. %
Recent works on efficient approximate sampling algorithms for the fixed-magnetization hard-core model \cite{Jain2022,Davies2023}, as well as works on computational and dynamical thresholds in the fixed-magnetization Ising model \cite{Carlson2022,Kuchukova2024}, rely on the concentration of the magnetization under the corresponding variable-magnetization models. Essentially, this allows the analysis of the fixed-magnetization model to be reduced to that of the variable-magnetization one. More general spin systems, like the GPM, do not satisfy this concentration property. %

The focus of our analysis is the fixed-magnetization version of the GPM (F-GPM), the Gibbs distribution of which is restricted to configurations with certain numbers of each color. To be precise, let~$\rhovec$ be a vector of $q$ positive numbers that sum to $1$. In other words, $\rhovec$ belongs to $\Delta_{q-1}$, the interior of the $(q-1)$-dimensional probability simplex. We define the set of configurations with color densities $\rhovec$ by
\[
\Omega_{G,\rhovec} := \left\{\sigma \in \Omega_G: \forall k \in [q-1],\,\, |\{v \in V(G): \sigma_v = k\}| = \lfloor \rho_k |V(G)| \rfloor \right\}.
\]
The F-GPM with magnetization $\rhovec$ is the probability distribution on $\Omega_{G,\rhovec}$, defined by
\begin{equation}\label{eq:pi fixed mag}
\pi_{G,\beta,\rhovec} (\sigma) := \pi_{G,\beta} \left( \sigma \mid \sigma \in \Omega_{G,\rhovec} \right) = \frac{e^{-\beta H_G(\sigma)}}{Z_G (\beta,\rhovec)}, \quad \text{where} \quad Z_G (\beta,\rhovec) := \sum_{\sigma \in \Omega_{G,\rhovec}} e^{-\beta H_G (\sigma)}.
\end{equation}
\cref{fig:examples1} provides examples of configurations sampled from $\pi_{G,\beta,\rhovec}$ with different cost matrices.

Unlike the preceding works on fixed-magnetization spin systems, we are primarily interested in the geometry of typical configurations under the fixed-magnetization Gibbs distribution at low temperature. To give a sense of the challenge this poses, we note that, even for the comparatively simple Ising model, the precisely characterizing the shape of the contours separating the $\pm 1$ spins under these conditions required a $200$-page tour de force by Dobrushin, Koteck{\'y}, and Shlosman \cite{Dobrushin1992}. The heft of their proof precludes its extension to more complicated spin systems, but the sophistication of its conclusion is fortunately unnecessary for applications to active programmable matter and computational models of biological processes. For example, a series of works use more fundamental analysis of contours, such as Peierls arguments and the cluster expansion, to analyze the fixed-magnetization Ising and $3$-state Potts models on~$\mathbb{Z}^2$ with unoccupied sites under certain connectivity requirements, motivated by their use in programmable matter \cite{Cannon2016,Cannon2019,Li2021}.

We now address the dynamics of the GPM and mixing. Applications to programmable matter typically use local updates, such as the Glauber or Kawasaki dynamics. The former is used for variable-magnetization models, while the latter is used for fixed-magnetization ones. Specifically, the {\it Glauber dynamics} changes the color of a single vertex at a time, so it does not preserve the density vector $\rhovec$. This is partly why the Hamiltonian of the CPM allows deviations in the numbers of vertices of each color, instead penalizing configurations for deviating too much from desired counts. Note that this choice of Hamiltonian prevents the dynamics from being implemented by a distributed algorithm. In contrast, the {\it Kawasaki dynamics} chooses an edge and exchanges the colors of the endpoints with appropriate probabilities, which preserves $\rhovec$. (This is the algorithm used to generate \cref{fig:examples1,fig:examples2}.) While there is reason to believe the Kawasaki dynamics converges rapidly on lattices, even at low temperature, rigorous polynomial-time bounds are elusive. They are at least as difficult as many notoriously challenging open problems, such as sampling configurations of the Ising model with homogeneous (say all $+1$) boundary conditions \cite{Martinelli1994} and, in the connected setting such as compression, lattice animals \cite{Rensburg1997}.  For this reason, we focus solely on understanding the geometric properties of typical configurations and leave the mixing time of these local algorithms as an open question.

\begin{figure}[t]
\centering
\includegraphics[width=.9\textwidth]{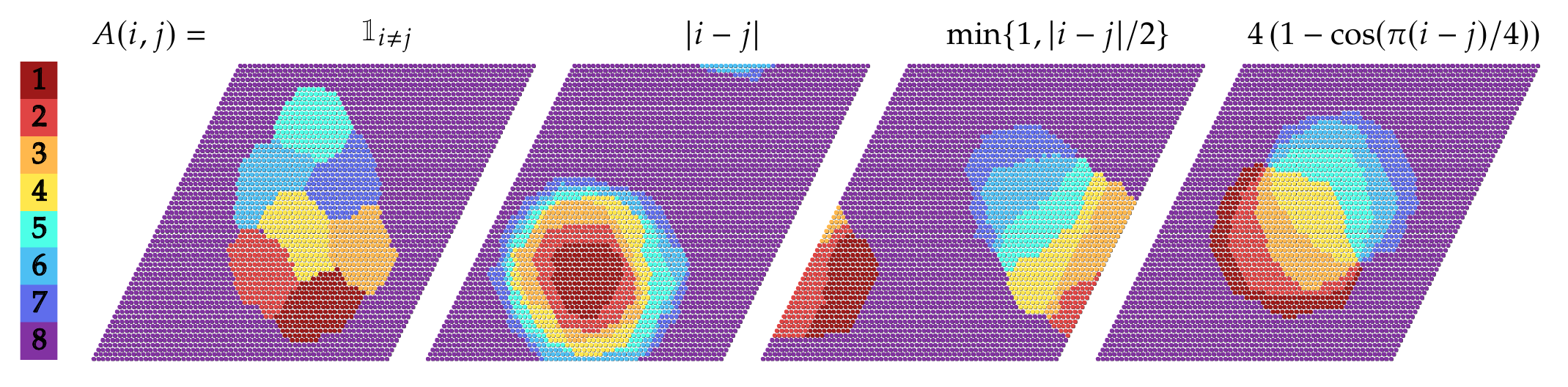}
\caption{Configurations of various $8$-color instances of the F-GPM on a $63 \times 63$ domain of the triangular lattice with periodic boundary conditions, $\beta = 1$, and density vector $\rhovec \propto (1,1,\dots,1,14)$. The first example is the Potts model, while the last is the clock model.}
\label{fig:examples1}
\end{figure}

\subsection{Results}\label{subsec:main result}

Our main result states that
at sufficiently low temperatures, a random configuration $\sigma$ under the Gibbs distribution $\pi_{G,\beta,\rhovec}$ of the F-GPM is sorted into regions that have low-cost (or short) boundaries and are nearly monochromatic, with high probability. We denote this event by $\mathsf{Sorted}(\alpha,\delta)$ in terms of parameters $\alpha$ and $\delta$ that control the extent to which the boundaries are short and the regions are monochromatic,  given $G$ and $\rhovec$. 

\newcommand{\partR}[1]{R_{#1}}
\newcommand{\partU}[1]{U_{#1}}
\newcommand{\partRall}{\partR{1},\partR{2},\ldots,\partR{q}}
\newcommand{\partUall}{\partU{1},\partU{2},\ldots,\partU{q}}
\newcommand{\partRshort}{\cal{R}}
\newcommand{\partUshort}{\cal{U}}
\newcommand{\partcost}[1]{\overline{H}_{#1}}
\begin{definition}[The $\mathsf{Sorted}$ event]\label{def:sorted}
For any partition ${\partUshort} = (U_i)_{i \in [q]}$ of the vertices $V(G)$, denote by $\partcost{G}(\partUshort)$ the Hamiltonian of a configuration where, for all $i \in \colorset$, the vertices in $\partU{i}$ have color $i$.
Then, for every density vector $\rhovec \in \Delta_{q-1}$, and parameters $\alpha > 1$ and $\delta \in (0,1)$, the event $\mathsf{Sorted} (\alpha,\delta)$ consists of configurations $\sigma \in \Omega_{G,\rhovec}$ that contain a partition ${\partRshort} = (R_i)_{i \in [q]}$ of $V(G)$ with the following properties.
\begin{enumerate}
 \item %
 For each $i \in \colorset$, there are at most $\delta |\partR{i}|$ vertices of colors other than $i$ within $\partR{i}$.
 \item %
 The partition $\partRshort$ has low energy, in the sense that
 \[
 \partcost{G}\left(\partRshort\right) \leq \alpha \cdot
 \min_{\substack{{\partUshort} \text{ partition of $V(G)$}\\\forall i \in [\numcolors-1],\, |\partU{i}| = \floor{\rho_i |V(G)|}}} \partcost{G}\left({\partUshort}\right).
 \]
\end{enumerate}
\end{definition}
The definition of $\mathsf{Sorted}$ generalizes the clustering property of \cite{Miracle2011} and separation in \cite{Cannon2019} to arbitrary numbers of colors and costs that are weighted contours or perimeters. 

We state the main theorem for an $L \times L$ region $\Lambda^L$ of the infinite triangular lattice $\Lambda^\infty$, where $L \geq 1$ is an integer side length and $\Lambda^L$ is equipped with periodic boundary conditions. We do so to facilitate the comparison of our results with those of previous works. Note that our approach can be easily adapted to $d$-dimensional square lattices, and likely to other classes of graphs as well. For convenience, we denote by $\mathcal{A}_q$ the set of cost matrices, i.e., symmetric matrices $A \in \R^{q\times q}$ such that $\mathrm{diag}(A) = \mathbf{0}$ and $A(i,j) > 0$ for all $i \neq j$. We further denote the smallest and largest off-diagonal entries of $A$ by $\costmatrixmin$ and $\costmatrixmax$ respectively.

\begin{theorem}[Main]\label{thm:main}
Let $q \geq 1$ be a number of colors, let $A \in \mathcal{A}_q$ be a cost matrix, and let $\rhovec \in \Delta_{q-1}$ be a density vector. For every $\alpha > 1$ and $\delta \in (0,1)$, 
 there exist values $\betaminF = \betaminF(\alpha,\delta,\rhovec,\numcolors,\costmatrixmin,\costmatrixmax)$ and $\LminF = \LminF(\beta,\alpha,\delta,\rhovec,\numcolors,\costmatrixmin,\costmatrixmax)$ such that, for any inverse temperature $\beta \geq \betaminF$ and domain side length $L \geq \LminF$, we have
 \begin{equation}\label{eq:main est}
 \pi_{\Lambda^L,\beta,\rhovec} \left(\mathsf{Sorted}(\alpha,\delta)\right) \geq 1 - e^{- \Omega \left( (\alpha-1) \beta L \right)}.
 \end{equation}
\end{theorem}

\cref{thm:main} provides the first results on the typical geometry of configurations under the fixed-magnetization Blume--Capel and $q$-state clock models, and it unifies some past results for other fixed-magnetization models, including the ferromagnetic Potts and Ising models. Additionally, it specializes to analogues of the main theoretical results of \cite{Cannon2016}, \cite{Cannon2019}, and \cite{Li2021}. These works used variants of the Ising and Potts models to design stochastic algorithms for compression, separation, and aggregation, respectively.

\begin{figure}[t]
\centering
\includegraphics[width=.9\textwidth]{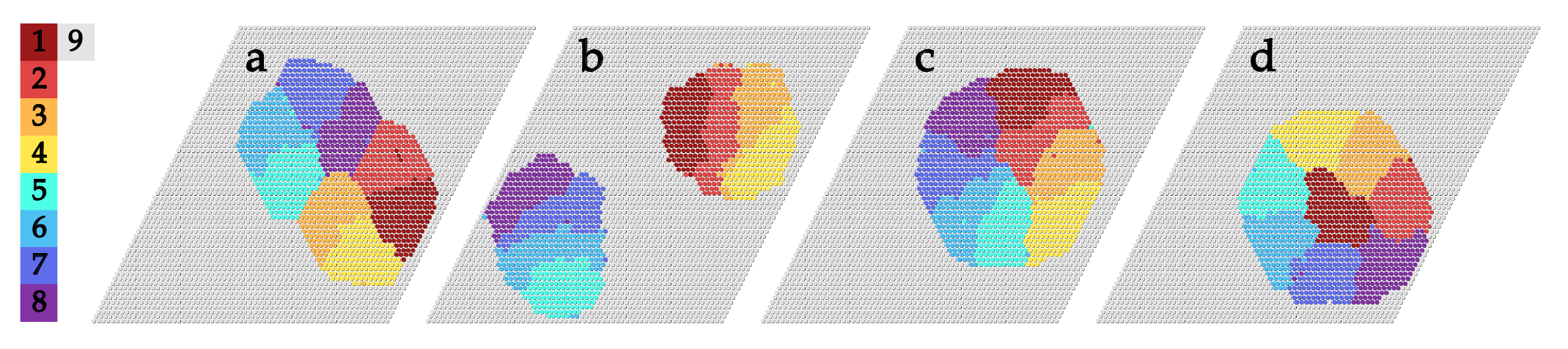}
\caption{Configurations of various $9$-color instances of the F-GPM on a $63 \times 63$ domain of the triangular lattice with periodic boundary conditions, $\beta = 1$, and density vector $\rhovec \propto (1,1,\dots,1,16)$. The cost matrices are listed in \cref{apx:interactionmatrices}. (a) Pairs of the form $(i,i+1)$ for $i \in \{1,3,5,7\}$ greatly attract, while the groups of $1$--$4$ and $5$--$8$ moderately attract. (b) Colors $1$--$4$ and $5$--$8$ are sorted, but these groups do not attract. (c) Like (b), except the groups attract. (d) Colors $2$--$8$ are incentivized to sort into a loop, while color $1$ aims to avoid color~$9$.}
\label{fig:examples2}
\end{figure}

Theorem~\ref{thm:main} allows us to {\it program} programmable matter by defining a suitable cost matrix $A$ and density vector $\rhovec$. Algorithmically, the Kawasaki dynamics allows pairs of neighboring particles to swap places, with transition probabilities given by the Metropolis--Hastings algorithm \cite{Metropolis1953}, so that the system provably converges to $\pi_{\Lambda^L,\beta,\rhovec}$. By Theorem~\ref{thm:main}, the particles self-organize into compact, nearly monochromatic regions, and the arrangements of these regions emerge from the off-diagonal entries of $A$. Figs.~\ref{fig:examples1} and~\ref{fig:examples2} give various examples for systems with $8$ or $9$ colors.

Lastly, a consequence of \cref{thm:main} of particular relevance to the CPM is that, when $\mathsf{Sorted}$ occurs, the parts $R_i$ of the implicit vertex partition $\mathcal{R} = (R_1, R_2, \dots, R_q)$ have sizes of roughly $\rho_i |V(G)|$. These sizes are analogous to the sizes of the ``cells'' in the CPM, which are instead enforced by an additional term in the Hamiltonian. We state this observation as a corollary.

\begin{corollary}
Suppose that $\sigma$ is a configuration in $\mathsf{Sorted} (\alpha,\delta)$ with the partition $\mathcal{R} = (R_1,R_2,\dots,R_q)$. Then, for each $i \in \colorset$, the size of $R_i$ satisfies
\begin{align*}
\frac{\rho_i - \delta}{1-\delta}|V(G)| - O(1) \leq |\partR{i}| \leq \frac{\rho_i}{1-\delta}|V(G)| + O(1),
\end{align*}
where the $O(1)$ term is due to rounding. Consequently, under the conditions of \cref{thm:main}, the probability of sampling a configuration that has this property satisfies the same bound as \cref{eq:main est}.
\end{corollary}

\begin{proof}
Denote by $n_i(\config)$ the number of vertices of color $i$ in configuration $\config$.
The first statement of \cref{def:sorted} tells us that at least $(1-\delta)|\partR{i}|$ vertices in $\partR{i}$ are of color $i$, and at most $\delta(|V(G)| - |\partR{i}|)$ vertices outside of $\partR{i}$ are of color $i$. Accordingly, $n_i (\sigma)$ satisfies
\begin{align*}
n_i(\config) &\geq (1-\delta)|\partR{i}|, \quad \text{and} \\
n_i(\config) &\leq |\partR{i}| + \delta(|V(G)| - |\partR{i}|) = (1-\delta)|\partR{i}| + \delta|V(G)|.
\end{align*}
Rearranging these terms and making use of the fact that $\left|n_i(\config) - \rho_i|V(G)|\right| \leq q$ (i.e., accounting for rounding error) completes the proof.
\end{proof}

\subsection{New proof techniques}\label{subsec:techniques}
The proof of Theorem~\ref{thm:main} centers on bounding above the contribution to the fixed-magnetization partition function from configurations that contain long contours (contours are defined precisely in Section~\ref{Sec:contours}). Ideally, we would use Pirogov--Sinai theory to represent this partition function as a contour polymer model and obtain the bound by controlling the corresponding cluster expansion~\cite{Friedli2017}. However, by default, cluster expansion techniques using the contour polymer model do not apply to the fixed-magnetization setting. Other works have avoided this problem by considering a variable-magnetization model that has highly concentrated magnetization that can be related to the corresponding fixed-magnetization within polynomial factors~\cite{Jain2022,Davies2023}. The general model that we study does not have this property. To overcome this barrier, we instead prove a sequence of bounds that go from the fixed-magnetization case to the variable-magnetization one, and back.

The virtue of this idea is that one of the intermediate bounds exclusively concerns the variable-magnetization case, and we address it by mostly following the preceding outline. However, due to the generality of our model, even proving the variable-magnetization bound is difficult. In the terminology of statistical physics, our model can have unstable ground states, the treatment of which requires a sophisticated version of Pirogov--Sinai theory involving truncations of the weights of contours \cite{Borgs2012,Borgs2020}. Furthermore, for the corresponding cluster expansions to converge, we must modify the Hamiltonian of our model with a magnetic field that balances the pressures associated with these truncated partition functions so that all states are simultaneously stable.

While we can address the variable-magnetization bounds through these nontrivial extensions of Pirogov--Sinai theory and cluster expansions, ``getting back'' to the fixed-magnetization case requires entirely new techniques. We complete this step by constructing a special contour set, which is short and contained by high-weight configurations with the correct magnetization. Interestingly, the special contour set arises as the contour set of the highest-weight configuration among configurations with a {\em different} density vector, arising from the inevitable perturbations of the density vectors by the most likely configurations of each sub-region. This density is related to the original one by a linear transformation that involves the truncated partition functions.

We view this technique as an extension of a Peierls argument in~\cite{Cannon2019}. A standard Peierls argument defines a mapping from configurations with long contours, or boundaries between regions of different colors, to those with shorter ones, to show that long contours are exponentially unlikely. The standard argument fails in our more general setting due to a significant change in the volume component of a configuration's total weight, which can easily outweigh the benefit of short contours. Instead, we map from classes of configurations with long contours to configurations containing our special contour set. By carefully managing the average change in the volume component across these configurations, we show that the cost of the contours ultimately plays the biggest role in which configurations are favored.

Our overall proof strategy, and the new techniques used to relate the variable- and fixed-magnetiza-tion partition functions, are likely to be valuable for future work on fixed-magnetization models and programmable matter. We discuss the strategy and techniques in greater detail in \cref{sec:proof ideas}.

\section{Overview of the proof structure}
\label{sec:proof ideas}
We provide a detailed overview of the proof, with an emphasis on the new techniques that we introduce. This section doubles as a guide to the rest of the paper, with forward references to the statements of key results and their proofs. Like Theorem~\ref{thm:main}, we specialize to the case of $G = \Lambda^L$. To keep the notation simple, we omit from quantities their dependence on the domain $\Lambda^L$ and its length $L$ whenever possible.

{\bf Disclaimer.} {\em This section introduces simplified versions of other important objects for the sake of exposition. We note where we omit technical aspects that do not contribute to conceptual understanding of the proof ideas, and include forward references to relevant parts of the proof. Lastly, we use different notation than in the actual proof, to emphasize only the most conceptually important aspects.}

Core to our analysis is the treatment of configurations as collections of contours. Informally, contours represent edges between regions of differently colored spins, which can be drawn as lines going between vertices of the lattice to indicate where the color changes are in the configuration. While this simple understanding of contours is all that is required to understand the structure of the proof, we elect to use the slightly more general construction of contours from Pirogov-Sinai theory, which is defined formally in \cref{Sec:contours}.

The proof entails analyzing the contributions of certain subsets of configurations to the variable- and fixed-magnetization partition functions. We define these subsets in terms of whether they contain a given contour set. Due to a minor technical point, when discussing the class of configurations that ``contain'' a contour set $\Gamma$, we in fact restrict this class to configurations that exclusively contain contractible contours, i.e., contours that do not ``wrap around'' $\regionL$ (Section~\ref{subsec:wrap}). For a contour set $\Gamma$, we define the set of such configurations that contain it as
\[
S(\Gamma) := \left\{\sigma \in \Omega: \Gamma(\sigma) \supseteq \Gamma \,\,\text{and}\,\, \text{every $\gamma \in \Gamma(\sigma)$ is contractible}\right\}.
\]
A convenient way to discuss contributions to the two partition functions is through the set function
\[
Z_\beta (S) := \sum_{\sigma \in S} e^{-\beta H(S)}, \quad S \subseteq \Omega.
\]
For example, $Z_\beta (S(\Gamma))$ is the contribution of $S(\Gamma)$ to the variable-magnetization partition function $Z_\beta (\Omega)$, while $Z_\beta (S(\Gamma) \cap \Omega_\rhovec)$ is the corresponding contribution to the partition function $Z_\beta (\Omega_\rhovec)$ with fixed magnetization $\rhovec$.

{\bf Key probability bound.} The key step in the proof of Theorem~\ref{thm:main} shows that, for a contour set~$\Gamma$, the probability under $\pi_{\beta,\rhovec}$ that a random configuration belongs to $S(\Gamma)$ is at most a number $\eps = \eps(\beta, |\overline\Gamma|) > 0$ that is exponentially small in $\beta$ and the contour's length $|\overline\Gamma|$. We can write this probability in terms of $Z_\beta$ as
\[
\pi_{\beta,\rhovec}(S(\Gamma)) = \frac{Z_\beta (S(\Gamma) \cap \Omega_\rhovec)}{Z_\beta (\Omega_\rhovec)},
\]
hence the bound $\pi_{\beta,\rhovec} (S(\Gamma)) \leq \eps$ is equivalent to
\begin{equation}\label{eq:idea1}
Z_\beta (S (\Gamma) \cap \Omega_\rhovec) \leq \eps \cdot Z_\beta (\Omega_\rhovec).
\end{equation}

We establish this bound in Lemma~\ref{lem:finalcontourbound}. While it does not immediately imply Theorem~\ref{thm:main}, we connect the two in the last step of the proof using a technique called {\em bridging} \cite{Miracle2011}. We discuss bridging at the end of this proof overview; for now, we focus on the probability bound.

Since \cref{eq:idea1} features the length of $\Gamma$ (through $\eps$), we would ideally establish it using Pirogov--Sinai theory, which entails representing a partition function as a sum over contours weighted by their lengths \cite{Pirogov1975,Pirogov1976,Friedli2017}. Indeed, this approach would suffice to establish the bound with $\Omega$ in the place of $\Omega_\rhovec$. However, it is inapplicable to fixed-magnetization partition functions, which, by default, have no contour polymer model. This is the first technical barrier to a proof of Theorem~\ref{thm:main}.

{\bf From fixed to variable magnetization, and back.} To overcome the first barrier, we split \cref{eq:idea1} into a sequence of four bounds that take us from the fixed-magnetization model to the variable-magnetization one, and back. Specifically, for a certain set of configurations $T \subset \Omega$, we essentially prove that
\begin{equation}\label{eq:strategy1}
Z_\beta (S(\Gamma) \cap \Omega_\rhovec) \stackrel{\#1}{\leq} Z_\beta (S(\Gamma)) \stackrel{\#2}{\leq} \eps (1-o(\beta)) \cdot Z_\beta (T) \stackrel{\#3}{\leq} \eps \cdot Z_\beta (T \cap \Omega_\rhovec) \stackrel{\#4}{\leq} \eps \cdot Z_\beta (\Omega_\rhovec).
\end{equation}
The following discussion focuses on bounds \#2 and \#3, as bounds \#1 and \#4 follow from the simple fact that if $S_1 \subseteq S_2 \subseteq \Omega$, then $Z_\beta (S_1) \leq Z_\beta (S_2)$. 

The virtue of our strategy is that bound~\#2 exclusively concerns variable-magnetization partition functions. Therefore, in principle, we can establish it by representing the partition functions as contour polymer models using Pirogov--Sinai theory, and then by controlling the resulting cluster expansions. However, doing so for our model is far more complicated than in previous works, because colors may not play symmetric roles in our model, in a sense. In statistical physics terminology, our model may have {\em unstable ground states}.

The treatment of unstable ground states requires a more sophisticated version of Pirogov--Sinai theory that works with {\em truncated partition functions} instead. (Indeed, past applications of Pirogov--Sinai theory have relied in a critical way on all ground states being stable \cite{Helmuth2019}, except for recent work on the Potts model \cite{Borgs2020}.) In our more general setting, the use of truncated partition functions introduces a further difficulty. For the truncated partition functions to have convergent cluster expansions, we must first modify the Hamiltonian with a magnetic field, to balance the {\em pressures} associated with these truncated partition functions. We address this in a moment, after briefly discussing bound~\#3.

Getting back to the fixed-magnetization partition function via bound~\#3 is difficult because it involves bounding a sum over configurations in $T$ by a sum over a subset of those configurations. In particular, this bound cannot hold if $Z_\beta (T)$ is dominated by configurations in $T$ that do not have magnetization $\rhovec$. It may seem as though the introduction of a magnetic field could help us, because it could increase the relative weight of configurations in $T$ with magnetization $\rhovec$. However, we are not free to choose the strength of the field, because bound~\#2 requires us to balance pressures first. This leaves the choice of the set $T$ as our only means of facilitating bound~\#3. But, for bound~\#2 to hold, $T$ must also contribute exponentially more to the partition function than $S(\Gamma)$ does. We must therefore choose the set $T$ carefully, to satisfy the demands of bounds~\#2 and \#3.

{\bf Introducing a magnetic field.} Our approach to establishing bound~\#2 requires us to modify the Hamiltonian with a magnetic field of a certain strength $\h \in \R^q$. (The strength is chosen in Lemma~\ref{lem:balance}, to balance the truncated pressures of the colors.) We define the modified Gibbs distribution and partition function by
\begin{equation}\label{eq:mag h and z}
\pi_\beta^\h(\sigma) \propto \exp \left(-\beta H (\sigma) + \h \cdot \mathbf{n} (\sigma) \right), \quad \sigma \in \Omega, \quad \quad \quad Z_\beta^{\hvec} (S) := \sum_{\sigma \in S} e^{-\beta H(\sigma) \, + \, \h\,\cdot\,\mathbf{n}(\sigma)}, \quad S \subseteq \Omega,
\end{equation}
where $\nfixedvec (\sigma) = (n_k (\sigma))_{k\in [q]}$ and $n_k (\sigma)$ is the number of vertices of color $k$ in $\sigma$.

An important observation is that, because the magnetic field applies only to the number of vertices of each color, it cannot affect the relative contributions of configurations to the fixed-magnetization partition function. To establish the key bound \cref{eq:idea1}, it therefore suffices to prove
\[
Z_\beta^{\h} (S (\Gamma) \cap \Omega_\rhovec) \leq \eps \cdot Z_\beta^{\h} (\Omega_\rhovec).
\]
We revise our existing strategy in \cref{eq:strategy1} accordingly, replacing partition functions of the form $Z_\beta (\cdot)$ with $Z_\beta^\h (\cdot)$.

{\bf Choosing the set $T$.} For bound~\#2 to hold, $T$ needs to contain configurations with contours that are significantly shorter than $\Gamma$, so that the corresponding configurations have exponentially greater weight in $|\overline\Gamma|$. For bound~\#3 to hold, $T$ needs to contain configurations of magnetization $\rhovec$ with relatively high weight among configurations in $\Omega_\rhovec$. We satisfy these demands by choosing $T = S(\Gamma_\ast)$, where $\Gamma_\ast$ is a special contour that is shorter than $\Gamma$ and is contained by configurations with high weight in $\Omega_\rhovec$.

We construct $\Gamma_\ast$ in \cref{subsec:construct good contours}, in terms of the field strength $\hvec$ that balances the pressures in the variable-magnetization model, as follows. Consider the matrix $\Theta \in \R^{q \times q}$ such that $\Theta(i,j)$ is the expected density $\Theta(i,j)$ of color $i$ in a ``region of mostly color $j$,'' under the distribution $\pi_{\beta}^{\hvec}$. (To define $\Theta (i,j)$ precisely requires truncated partition functions.) We transform the density vector~$\rhovec$, which we treat as a column vector, according to $\rhovec_\ast = \Theta^{-1} \rhovec$. We then define $\sigma_\ast$ to be the configuration in $\Omega_{\rhovec_\ast}$ with the least cost, i.e., such that $H(\sigma_\ast) \leq \min_{\config \in \Omega_{\rhovec_\ast}} H(\config)$. Lastly, we define $\Gamma_\ast$ to be the contour set of $\sigma_\ast$.

{\bf Revised sequence of bounds.} We revise \cref{eq:strategy1} based on the preceding discussion. For future reference, we include more precise forms of the factors involving $\beta$. To prove the key probability bound, we establish a sequence of inequalities like:
\begin{align}\label{eq:strategy2}
Z_\beta^\h (S(\Gamma) \cap \Omega_\rhovec) &\leq Z_\beta^\h (S(\Gamma)) \nonumber\\ 
&\leq \exp(-\beta |\overline\Gamma| + \beta |\overline\Gamma_\ast|) \cdot Z_\beta^\h (S(\Gamma_\ast)) \nonumber\\ 
&\leq \underbrace{\exp(-\beta |\gamma| + \beta |\Gamma_\ast|) \cdot \exp(o(\beta)|\Gamma_\ast|)}_{\eps(\beta,|\overline\Gamma|)} \,\cdot\, Z_\beta^\h (S(\Gamma_\ast) \cap \Omega_\rhovec) \leq \eps \cdot Z_\beta^\h (\Omega_\rhovec).
\end{align}
The final sequence of inequalities appears in the proof of Lemma~\ref{lem:finalcontourbound}. There, the second bound involves the cost $\mathcal{H}(\Gamma)$ of the contour set, instead of its length. However, the length and cost are within constant factors of one another.

{\bf Bridging.} In the last step of the proof of Theorem~\ref{thm:main}, we connect the probability bound from \cref{eq:strategy2} to the theorem's conclusions using the technique of bridging~\cite{Miracle2011,Cannon2019,Kedia2022}. For any configuration $\sigma \in \Omega_\rhovec$, bridging identifies a subset $\confinal$ of the configuration's contour set $\Gamma(\sigma)$ that meets the second property of $\mathsf{Sorted}$, namely, that each component of $V(\regionL) \setminus \contoursetbar$ is monochromatic, aside from at most a fraction $\delta \in (0,1)$ of vertices. What makes bridging useful is that we can bound above the number of contour subsets with a given length or cost that it can produce. Then, so long as $\eps$ decays rapidly enough in $|\overline\Gamma|$, we can conclude that configurations for which bridging produces a long or costly contour set are rare.

We introduce bridging and state its properties in Lemma~\ref{lem:bridgesystems}. We bound the number of contour subsets of a given length or cost that bridging can produce in Lemma~\ref{lem:numbridgesystems}. The proof of Theorem~\ref{thm:main} that follows these statements combines them with the key probability bound from Lemma~\ref{lem:finalcontourbound}.

%% file: proof_section.tex
\section{Preliminaries for the contour polymer model}
\label{Sec:contours}
We now provide some of the formalism that is required for the complete proofs just outlined in \cref{sec:proof ideas}. 
A {\em contour} $\gamma$ is a pair $(\overline\gamma, \omega_{\overline\gamma})$ consisting of a finite support $\overline\gamma$, which is a connected set of vertices in $V(G)$, and an assignment $\omega_{\overline\gamma}: \overline\gamma \to [q]$ of colors to the vertices of $\overline\gamma$. We define the {\em contour set} $\Gamma (\sigma)$ of a configuration $\sigma \in \Omega_G$ in the following way. First, we define the set of vertices that participate in bichromatic edges as
\[
\overline\Gamma (\sigma) := \left\{u \in V(G): \exists v \sim u: \sigma_u \neq \sigma_v \right\},
\]
where the notation $v \sim u$ indicates that $(u,v)$ is an edge of $G$. Second, we consider the connected components $\overline\gamma_1$ through $\overline\gamma_k$ of the subgraph $G(\overline\Gamma (\sigma))$ induced by restricting $G$ to the vertices of $\overline\Gamma (\sigma)$, with assignments that agree with $\sigma$:
\[
\omega_{\overline\gamma_i} (v) := \sigma_v, \quad v \in \overline\gamma_i.
\]
We then define the contour set of $\sigma$ to be the union of contours $\gamma_i = (\overline\gamma_i,\omega_{\overline\gamma_i})$:
\[
\Gamma (\sigma) := \left\{\gamma_i: i \in [k]\right\}.
\]
Lastly, we define an analogue of the Hamiltonian that quantifies the ``cost'' of a contour set. For a contour set $\Gamma$, we define
\[
\mathcal{H}_G (\Gamma) := \sum_{\gamma \in \Gamma} \, \sum_{(u,v) \in E \left( G (\overline \gamma) \right)} A \left( \omega_{\overline\gamma} (u), \omega_{\overline\gamma} (v) \right).
\]
For convenience of notation, if $\gamma$ is an individual contour, we use the notation $\mathcal{H}_G (\gamma)$ to denote $\mathcal{H}_G(\{\gamma\})$. Note that $\mathcal{H}_G (\Gamma (\sigma)) \leq H_G (\sigma)$ by definition. For the rest of the proof, we will also drop $G$ from the notation as the underlying graph is clear.

As we would like to be able to compare contours between different sizes $L$ of the lattice $\regionL$, we use the infinite triangular lattice $\fulllattice$ as a baseline. Denote by $\allcontours$ the set of all possible contours on~$\fulllattice$. Note that we only consider a pair $(\contourbar, \contourassg{\contourbar})$ a contour if it can be constructed (in the same way as before) from assignments $\infassg$ of colors to $V(\fulllattice)$ such that there exists a color $j$ where only a finite number of sites are assigned colors other than~$j$. Note that this means each contour must have a finite support $\contourbar$.

\subparagraph{Labels and Compatibility.}
The assignment $\contourassg{\contourbar}$ of a contour $\contour$ must have the property that for each component $R$ of $V(\fulllattice) \setminus \contourbar$, every vertex of $\contourbar$ neighboring a vertex in $R$ must be assigned the same color. This gives us a label in $\colorset$ for each of these components. 
As $\contourbar$ is finite, there is one infinite component $\extr(\contour)$ of $V(\fulllattice) \setminus \contourbar$, which we call the exterior. The label of the exterior is what we call the \emph{type} of the contour $\contour$.
Furthermore, for each $j \in \colorset$, we define $\intr{j}(\contour)$ to be the union of the non-exterior components that are of type $j$ (Figure~\ref{fig:contourdiag}).

We say that two contours $\contour_1$, $\contour_2$ are \emph{compatible} if and only if $d_{\fulllattice}(\contourbar_1, \contourbar_2) \geq 2$, in terms of the graph distance $d_{\fulllattice}$ on $\fulllattice$. This notion of compatibility is relevant in our polymer model which is used to show that the partition functions can be expressed as convergent cluster expansions.

On the other hand, we say that a set of contours $\contourset$ defined over a graph $G$ is \emph{consistent} if the contours in the set are pairwise-compatible and for each component $U$ of $V(G)\setminus \contoursetbar$ (where $\contoursetbar = \bigcup_{\contour \in \contourset}\contourbar$), all vertices of $\partialext U$ have the same color.

\begin{figure}[!t]
\centering
\begin{tikzpicture}[x=0.6cm,y=0.6cm]
\input{diagrams/contourdiag}
\end{tikzpicture}
\caption{A contour $\contour = (\contourbar, \contourassg{\contourbar}) \in \allcontours_0$, defined over the infinite lattice $\fulllattice$. The contour divides the lattice into $\contourbar$, the exterior $\extr(\contour)$ (of type $1$) and the interiors $\intr{i}(\contour), i \in \colorset$. The assignment $\contourassg{\contourbar}$ applies colors to each vertex in $\contourbar$, which are shown in red. These colors dictate the label of each component of the subgraph of $\fulllattice$ induced by $V(\fulllattice)\setminus \contourbar$.}
\label{fig:contourdiag}
\end{figure}

\subparagraph{Sufficiently low temperatures.}
Most of the Lemmas that we will show over the course of this proof will rely on the inverse temperature $\beta$ and region side length $L$ being sufficiently large. To ensure that there is a single lower bound for $\beta$ and $L$ for which all lemmas hold, we pre-define constants $\betamin = \betamin(\rho,q,\costmatrixmin,\costmatrixmax)$ and $\Lmin(\beta) = \Lmin(\beta,\rho,q,\costmatrixmin,\costmatrixmax)$, and show that our statements hold for all $\beta \geq \betamin$ and $L \geq \Lmin(\beta)$. To determine the values of $\betamin$ and $\Lmin(\beta)$ and the bounds they need to satisfy, we refer the reader to Appendix~\ref{apx:allconstants}. Note however that most of these bounds have been intentionally made loose for simplicity. A proper discussion of the minimum value of $\beta$ and $L$ for such bounds to hold is outside the scope of this paper. In addition, some of our lemmas depend on additional parameters $\alphaA$ or $\alphaB$. We define $\betaminA(\alphaA)$ and $\betaminB(\alphaB)$ similarly for these purposes in Appendix~\ref{apx:allconstants}.

\subparagraph*{Partition Functions.}
For a finite region $V \subseteq V(\fulllattice)$, we denote by $\CFGFint{j}{V}$ the set of configurations on $V$ (as assignments of the particles of $V$ to colors in $\colorset$) that have contour configurations that are surrounded by color $j$ and have no contour that touches the internal boundary of $V$. This definition comes from Section~7.3 of~\cite{Friedli2017} and is used for the recursive construction of partition functions. To be precise,
\[
\CFGFint{j}{V} = \left\{ \infconfig \in \CFGF: \infconfig(v) = j \text{ for all } v \text{ where } d_{\fulllattice}(v,V^\compl) \leq 2 \right\}.
\]
The partition function for $\CFGFint{j}{V}$ is then given as:
\[
\ZFbd{j}(\beta, \hvec, V) = \sum_{\config \in \CFGFint{j}{V}} \exp\left(-\beta H (\sigma) + \h \cdot \mathbf{n} (\sigma)\right).
\]
We use similar definitions if $V \subseteq V(\regionL)$ instead (where $\regionL$ is the $L \times L$ sub-lattice of the triangular lattice with periodic boundary conditions), replacing $\CFGFint{j}{V}$ with $\CFGPint{L}{j}{V}$ and $\ZFbd{j}(\beta, \hvec, V)$ with $\ZPbd{L}{j}(\beta, \hvec, V)$.

\subparagraph*{Contour Weights.}
For $j \in \colorset$, denote by $\allcontoursbd{j}$ the set of contours of $\fulllattice$ of type $j$ (exterior boundary has label $j$). The weight $\wutr{j}(\contour)$ of a contour $\contour = (\contourbar, \contourassg{\contourbar}) \in \allcontoursbd{j}$ is given by the following recursive definition:
\begin{align*}
\wutr{j}(\contour) = e^{\contourcost{\contour}} \prod_{j' \in \colorset} \frac{\ZFbd{j'}(\beta, \hvec, \intr{j'}(\contour))}{\ZFbd{j}(\beta, \hvec, \intr{j'}(\contour))},
\text{ where }
\contourcost{\contour} = -\beta \Hcost{\contour} + \sum_{i \in \colorset}(h_i-h_j)n_i(\contourassg{\contourbar}).
\end{align*}
Note that the weight $\wutr{j}(\contour)$ depends on both the inverse temperature $\beta$ and the magnetic field $\hvec$.
As can be seen from Section 7.3.1 of~\cite{Friedli2017}, this definition of $\wutr{j}$ allows us to express the partition function $\ZFbd{j}(\beta, \hvec, V)$ in terms of a polymer model that uses only contours of type $j$, even if the interiors of the contours have different labels.

Denoting by $\contourcfgs_j(V)$ the set of mutually compatible contours in $\allcontoursbd{j}$ that are fully contained within $V \setminus \partialint(V)$ (where $\partialint(V)$ is the interior boundary of the vertex set $V$), this gives us the following alternative expression for the partition function:
\begin{align*}
\ZFbd{j}(\beta, \hvec, V) = e^{h_j |V|} \sum_{\contourset \in \contourcfgs_j(V)} \prod_{\contour \in \contourset} \wutr{j}(\contour).
\end{align*}
Being able to express our partition function with a polymer model allows us to write it as a cluster expansion, which will be used heavily for our later proofs. Note that this partition function is only used for illustration of the purpose of these contours. While these contours are defined on the infinite lattice, the main partition functions we will be analyzing will be in the periodic setting (\cref{lem:nonwrappingpartitionfunction}, which we will prove).

\subparagraph*{Truncated weights.}
Truncated weights are an important tool for our proofs, as they allow us to analyze the partition functions relating to individual ground states. We use the definition of truncated weights in~\cite{Borgs1990}; $\wtr{j}$ is a translation-invariant function that maps contours to nonnegative real values, just like $\wutr{j}$. However, the truncated weight differs by the addition of an additional factor, which truncates contours by giving them a weight of $0$ if their presence causes too significant impact on the partition function of their interior.
The truncated weight $\wutr{j}(\contour)$ of a contour $\contour$ and the truncated partition function $\ZFtr{j}(V)$ of a region $V \subseteq \fulllattice$ are defined together inductively on the size of the region $V$, as their definitions rely on each other.
The truncated weight of a contour $\contour$ is defined as:
\begin{align*}
\wtr{j}(\contour) := e^{\contourcost{\contour}} \chi\left(
\log \ZFtr{j}(V_\contour) -\log \ZFtr{j'}(V_\contour) + \beta \frac{\costmatrixmin}{16}|\contourbar|
\right) \prod_{j' \in \colorset} \frac{\ZFbd{j'}(\beta, \hvec, \intr{j'}(\contour))}{\ZFbd{j}(\beta, \hvec, \intr{j'}(\contour))},
\end{align*}
where $V_\contour = \contourbar \cup \bigcup_{i \in \colorset}\intr{i}(\contour)$ and $\chi$ is a smooth threshold function such that $\chi(x) = 0$ for $x \leq -1$ and $\chi(x) = 1$ for $x \geq 1$.
The truncated partition function of a region $V$ is defined in the same way as the untrucated partition function, but using the truncated weight in place of the contour weight:
\begin{align*}
\ZFtr{j}(\beta, \hvec, V) = e^{h_j |V|} \sum_{\contourset \in \contourcfgs_j(V)} \prod_{\contour \in \contourset} \wtr{j}(\contour).
\end{align*}
Because the precise definition of truncated weights is not relevant to our proofs, we defer the reader to either~\cite{Borgs1990} or~\cite{Friedli2017} for a detailed discussion of truncated weights.

We next define truncated pressures, which represent per-vertex contributions to the partition function of a region of ``mostly $j$'', for any $j \in \colorset$. In a configuration sampled from a region of the lattice, the colors with the highest pressures tend to dominate with high probability, which makes it difficult to analyze the contributions of the remaining colors to the partition function.
Thus, to be able to properly analyze a fixed-magnetization partition function with different colors dominating in different regions of the lattice, we make use of a specific choice of magnetic field $\hvecopt$ that equalizes the pressures of all colors. This is done in \cref{lem:balance}.

\begin{definition}[Truncated Pressure]
For each $j \in \colorset$ and magnetic field strength $\hvec \in \R^q$, we define:
\begin{align*}
\pressuretr{j}(\hvec) := \lim_{L \to \infty}\frac{1}{|V(\Lbox)|}\log \ZFtr{j}(\beta, \hvec, V(\Lbox)).
\end{align*}
\end{definition}

\subparagraph*{Stability of Weights.}
The terms of this partition function (as well as others we will later define and use) can be formally rearranged into a sum over clusters (a cluster expansion), which will be important for many of our proofs.
However, for the cluster expansion to converge, we would want our weights to be \emph{$\tau$-stable} for some $\tau > 0$. We present only the idea for the arguments, and refer the reader to~\cite{Borgs1990} and~\cite{Friedli2017} for a more thorough analysis.
\begin{definition}[$\tau$-stable]
A translation-invariant function $w$ mapping contours to nonnegative weights is $\tau$-stable if for all contours $\contour$ it satisfies
\[
w(\contour) \leq e^{-\tau|\contourbar|}.
\]
\end{definition}

Another important concept the balancing of pressures using a specific choice of magnetic field $\hvec$. Following the analysis seen in~\cite{Borgs1990,Friedli2017}, we can always find a choice of magnetic field $\hvecopt$ not too far from $\mathbf{0} \in \R^{\numcolors}$ such that all of pressures are equal. A pressure represents the contribution to the partition function per unit volume of a region, and so ground states corresponding to the largest pressures tend to dominate.
\cref{lem:balance} establishes this, and defines a small open set $\hoptopenset'$ around $\hvecopt$ within which the stability bounds hold.
\begin{lemma}[Balancing Pressures and Stability]
\label{lem:balance}
If $\beta \geq \betamin$ and $L \geq \Lmin(\beta)$, then there exists a magnetic field strength $\hvecopt \in \R^{q}$ such that $\|\hvecopt\|_{\infty} \leq \hoptbound$ and
\begin{align*}
\pressuretr{1}(\hvecopt) = \pressuretr{2}(\hvecopt) = \cdots = \pressuretr{\numcolors}(\hvecopt).
\end{align*}
Additionally, there exists an open set $\hoptopenset' \subseteq \R^q$ containing $\hvecopt$ such that for all $\hvec \in \hoptopenset$, colors $i,j \in \colorset$, and $k \in \derivset$, the weights satisfy
\[
\wutr{j}(\contour) \leq e^{-\beta\hetcost |\contourbar|/8}
\quad\text{and}\quad
\derivki{k}{i} \wtr{j}(\contour) \leq e^{-\beta\hetcost |\contourbar|/8}.
\]
For the rest of the paper, we will denote $\taub = \beta \hetcost / 8$, so that these weights are $\taub$-stable.
\end{lemma}

\begin{proof}
The stability of the weights and the derivatives of the truncated weights within a small region $\hoptopenset'$ around $\mathbf{0}$ come from Lemmas 2.1 and 2.3 of~\cite{Borgs1990}.
By the same paper, $\hvecopt$ exists and is an interior point of this region. We make an additional argument below, following the analysis in~\cite{Friedli2017} to explain the upper bound for $\|\hvecopt\|_\infty$.

As the truncated weights $\wtr{j}(\contour)$ are $\taub$-stable, the cluster expansion 
$\sum_{\cluster \in \clustersetbd{j}} \clwtr{j}(\cluster)$ converges by \cref{lem:convergencecriterion} for $\beta \geq \betamin$ (Appendix~\ref{apx:allconstants}), gives us the following expression, where we denote $V = V(\regionL)$ for brevity.
\begin{align*}
\log \ZFtr{j}(\beta, \hvec, V(\Lbox))
&= h_j|V| + \sum_{\substack{\cluster \in \clustersetbdv{j}{V(\regionL)}}} \clwtr{j}(\cluster) \\
&= h_j|V| +\sum_{x \in V} \sum_{\substack{\cluster \in \clustersetbd{j}, \clusterbar \ni x}} \frac{1}{|\clusterbar|} \clwtr{j}(\cluster)
- \sum_{x \in V} \sum_{\substack{\cluster \in \clustersetbdv{j}{V(\regionL)}, \clusterbar \ni x \\ \clusterbar \cap \partialint V \neq \emptyset}} \frac{1}{|\clusterbar|} \clwtr{j}(\cluster) \\
&= h_j|V| +|V| \sum_{\substack{\cluster \in \clustersetbd{j}, \clusterbar \ni 0}} \frac{1}{|\clusterbar|} \clwtr{j}(\cluster)
- \sum_{\substack{x \in \partialint V}} \sum_{\substack{\cluster \in \clustersetbdv{j}{V(\regionL)}, \clusterbar \ni x \\ \clusterbar \cap (V \setminus \partialint V) \neq \emptyset}} \frac{1}{|\clusterbar \cap \partialint V|} \clwtr{j}(\cluster),
\end{align*}
This gives us the following approximation for the log partition function:
\begin{align*}
\left|\frac{1}{|V(\regionL)|}\log \ZFtr{j}(\beta, \hvec, V(\Lbox)) - h_j - \sum_{\contour \in \clustersetbd{j}, \clusterbar \ni 0} \frac{1}{|\clusterbar|}\clwtr{j}(\contour) \right|
\leq |\partialext V(\regionL)| e^{-\taub \minlength/4}.
\end{align*}
Taking the limit as $L \to \infty$, we thus get the following expression for the pressure:
\begin{align*}
\pressuretr{j}(\hvec) = h_j + \sum_{\cluster \in \clustersetbd{j}, \clusterbar \ni 0}\frac{1}{|\clusterbar|} \clwtr{j}(\cluster),
\text{ where }
\left|\sum_{\cluster \in \clustersetbd{j}, \clusterbar \ni 0}\frac{1}{|\clusterbar|} \clwtr{j}(\cluster) \right| \leq e^{-\taub \minlength/4}.
\end{align*}
What this allows us to conclude is that for any $j \in \colorset$, if $\ds h_j \geq 3e^{-\taub \minlength/4} + \max_{i \neq j} h_i$, we must have $\ds \pressuretr{j}(\hvec) > \max_{i \neq j}\pressuretr{i}(\hvec)$.
Thus, as $\hvecopt$ is where the pressures are equalized, for each $j \in \colorset$, there must be another color $j' \neq j$ where $|\hoptj{j'} - \hoptj{j}| \leq 3e^{\taub \minlength/4}$. Consequently, as magnetic field strengths only matter relative to each other, we may set $h_1 = 0$ and observe that this necessarily implies that $\|\hvecopt\|_\infty \leq \hoptbound$.
\end{proof}
\cref{lem:truncatedequal} then establishes the primary reason for why we are interested in the derivatives of the truncated partition functions; it allows us to analyze the same derivatives of the untruncated partition functions at $\hvecopt$ with the truncated partition functions under the same boundary conditions, as they are equal within another open set $\hoptopenset''$ containing $\hvecopt$.
\begin{lemma}[Equivalence of truncated and untruncated partition functions]
\label{lem:truncatedequal}
There exists an open set $\hoptopenset''$ containing $\hvecopt$ such that for all $\hvec \in \hoptopenset''$, colors $i, j \in \colorset$ and $k \in \derivset$, we have
\begin{align*}
\derivki{k}{i} \wtr{j}(\contour) 
&= \derivki{k}{i}  \wutr{j}(\contour).
\end{align*}
for all contours $\contour \in \allcontoursbd{j}$.
Consequently, truncated and untruncated partition functions over any region $V \in V(\fulllattice)$, which are defined as sums of contour weights, will be equal:
\begin{align*}
\derivki{k}{i} \log \ZPtr{L}{j}(\beta, \hvecopt, V) 
&= \derivki{k}{i}  \log \ZPbd{L}{j}(\beta, \hvecopt, V).
\end{align*}
\end{lemma}

\begin{proof}
This comes from Lemma 2.1(iii) of~\cite{Borgs1990}. As $\pressuretr{j}(\hvecopt) - \min_{j'}\{\pressuretr{j'}(\hvecopt)\} = 0$ for each $j \in \colorset$, according to the lemma, there is an open set in $\R^q$ containing $\hvecopt$ such that for all $\hvec$ in this set, we have
\begin{align*}
\wtr{j}(\contour) = \wutr{j}(\contour).
\end{align*}
The equality of the derivatives within this set thus follows.
\end{proof}
Our bounds later on will focus specifically on $\hvecopt$, but the open set $\hoptopenset = \hoptopenset' \cap \hoptopenset''$ is necessary for the discussion of derivatives of the partition functions.
For the rest of this paper, we only consider magnetic fields $\hvec$ within $\hoptopenset$.

\subsection{Non-contractible contours in the periodic setting}\label{subsec:wrap}
One issue when working with periodic boundaries is that contours do not have a well-defined exterior. To ensure that the properties of truncated weights continue to hold in the periodic setting, we analyze our partition functions while omitting contours which have exteriors that do not coincide with their exterior if they were moved to the infinite lattice $\fulllattice$.
\begin{definition}
Given a finite periodic sub-lattice $\regionL$ and a non-empty subset $V$ of its vertices, a contour in $\regionL$ is \emph{non-contractible} in $V$ if it contains a closed path that cannot be contracted to a single point without increasing in length or crossing $V^\compl$. Otherwise, we call the contour \emph{contractible}.
\end{definition}

For any contour that is contractible in the region $V$, if it moved to the infinite lattice $\fulllattice$, it will have an exterior that coincides with the component of $V(\regionL) \setminus V$ that contains the complement $V^\compl$.
For the rest of this paper, for any weight function $w: \allcontoursbd{j} \to \R_{\geq 0}$ which we assume to be translation-invariant, we may apply the same weight function to any contour $\contour \in \CFGPSint{L}{j}{V}$, by ``unwrapping'' $\contour$ to be a contour in the infinite lattice $\fulllattice$, and taking its weight under $w$. This can be done because of the translation invariance of $w$.

When we analyze partition functions in the finite periodic lattice $\regionL$, we often look at $\CFGPSint{L}{j}{V}$, the set of configurations on the lattice that contain only contours that are contractible in $V$,
as opposed to the set $\CFGPint{L}{j}{V}$ of all configurations on $V$ in the lattice $\regionL$.
The partition function corresponding to $\CFGPSint{L}{j}{V}$ is $\ZPSbd{L}{j}(\beta, \hvec, V)$.
\begin{lemma}[Partition function of contractible contours]
\label{lem:nonwrappingpartitionfunction}
Denoting by $\contourcfgsperbdv{L}{j}{V}$ the set of all mutually compatible contours of type $j$ in $\regionL$ that are contractible in $V$, we have
\begin{align*}
\ZPSbd{L}{j}(\beta, \hvec, V) = e^{h_j |V|} \sum_{\contourset \in \contourcfgsperbdv{L}{j}{V}} \prod_{\contour \in \contourset} \wutr{j}(\contour).
\end{align*}
\end{lemma}

\begin{proof}
Each configuration $\config$ of $\CFGPSint{L}{j}{V}$ can be divided into contours $\contourset(\config)$. Denoting by $\contoursetext(\config)$ the set of contours in $\contourset(\config)$ that are not within any interior component of any other contour in $\contourset(\config)$, we have by summing over all possible values of $\contoursetext$ and by induction on the size of $V$:
\begin{align*}
\ZPSbd{L}{j}(\beta, \hvec, V)
&= e^{h_j|V|} \sum_{\contoursetext} \prod_{\contour \in \contoursetext} \left( e^{\contourcost{\contour}} \prod_{j' \in \colorset} \left(e^{ -h_j|\intr{j'}(\contour)|}\ZPSbd{L}{j'}(\beta,\hvec,\intr{j'}(\contour))\right)\right) \\
&= e^{h_j|V|} \sum_{\contoursetext} \prod_{\contour \in \contoursetext} \left( e^{\contourcost{\contour}} \prod_{j' \in \colorset} \left(\frac{\ZPSbd{L}{j'}(\beta,\hvec,\intr{j'}(\contour))}{\ZPSbd{L}{j}(\beta,\hvec,\intr{j'}(\contour))}e^{ -h_j|\intr{j'}(\contour)|}\ZPSbd{L}{j}(\beta,\hvec,\intr{j'}(\contour))\right)\right) \\
&= e^{h_j|V|} \sum_{\contoursetext(\config)} \prod_{\contour \in \contoursetext(\config)} \left( \wutr{j}(\contour) \prod_{j' \in \colorset} \left(e^{ -h_j|\intr{j'}(\contour)|}  \cdot  e^{h_j |\intr{j'}(\contour)|} \sum_{\contourset \in \contourcfgsperbdv{L}{j}{\intr{j'}(\contour)}} \prod_{\contour \in \contourset} \wutr{j}(\contour)  \right)\right) \\
&= e^{h_j |V|} \sum_{\contourset \in \contourcfgsperbdv{L}{j}{V}} \prod_{\contour \in \contourset} \wutr{j}(\contour). \qedhere
\end{align*}
\end{proof}

We focus on partition functions with contractible contours as we analyze configurations by splitting contours of configurations into those that ``wrap'' and those that do not. To state this clearly, we make the following definition:

\begin{definition}[Closure under contractibility]
\label{defn:closedunderwrapping}
Given a configuration $\config$ on $\regionL$ and a subset $\conbad$ of the contours of $\config$, we say that $\conbad$ is closed under contractibility in $\config$ if no component of $V(\regionL) \setminus \conbadbar$ (where $\conbadbar = \bigcup_{\contour \in \conbad} \contourbar$) contains a contour of~$\config$ that is non-contractible in it.
For a set of contours $\conbad$, we denote by $\CFGPc{L}{\conbad}$ the set of configurations $\config$ which contain the contours in $\conbad$ and where $\conbad$ is closed under contractibility in $\config$.
\end{definition}
We note that a set of contours that is closed under contractibility in some configuration must necessarily be consistent.
Thus, if $V_1(\conbad), V_2(\conbad), \ldots, V_K(\conbad)$ are the components of $V(\regionL) \setminus \conbadbar$, the contour set $\conbad$ being closed under contractibility means that each one of these components $V_k$ will have a well-defined label $j_k$, in the sense that every vertex of $\conbadbar$ adjacent to a vertex in $V_k$ will have color $j_k$.
This allows us to write the partition function corresponding to $\CFGPc{L}{\conbad}$ as:
\begin{align}
\label{eqn:wrapexpansion}
\ZPc{L}{\conbad}(\beta, \hvec) = e^{\sum_{\contour \in \conbad} \contourcost{\contour}} \prod_{k=1}^K \ZPSbd{L}{j_k}(\beta, \hvec, V_k).
\end{align}

\section{Proof of \cref{thm:main}}
In the proof of \cref{thm:main}, we compare the partition functions corresponding to some consistent contour set $\conbad$ on $\regionL$ to that of an ``optimal'' contour set $\congood$ that we will construct.
Let $\CFGPcf{L}{\conbad}{\nfixedvec}$ and $\CFGPcf{L}{\congood}{\nfixedvec}$ denote the fixed-magnetization versions of the sets of configurations $\CFGPc{L}{\conbad}$ and $\CFGPc{L}{\conbad}$ corresponding to $\conbad$ and $\congood$ respectively as given in \cref{defn:closedunderwrapping}, i.e., with the restriction that the number of vertices of color $i$ is $\nj{i}$ for each $i \in \colorset$.

Denote by $\ZPcf{L}{\conbad}{\nfixedvec}(\beta,\hvec)$ the partition function for $\CFGPcf{L}{\conbad}{\nfixedvec}$ at the fixed magnetization $\nfixedvec$, while $\ZPc{L}{\conbad}(\beta, \hvec)$ denotes the partition function for $\CFGPc{L}{\conbad}$ without the restriction of fixed magnetization.
The approach to the proof is to establish the four inequalities from \cref{eq:strategy1}. The first of the four is the following inequality for the ``bad contours'' $\conbad$, which easily applies for any choice of inverse temperature~$\beta$ as $\CFGPcf{L}{\conbad}{\nfixedvec}$ is a subset of $\CFGPc{L}{\conbad}$:
\begin{align}
\label{eqn:badineq}
\ZPcf{L}{\conbad}{\nfixedvec}(\beta,\hvecopt) \leq \ZPc{L}{\conbad}(\beta, \hvecopt).
\end{align}
The bulk of the proof will be covering the second and third inequalities. However, before this can be done, fundamentals about cluster expansions, the grand canonical ensemble and $\congood$ need to be established.
\subsection{Adjusted densities and constructing $\congood$}\label{subsec:construct good contours}
In the grand canonical ensemble (no restrictions on $n_0(\config),n_1(\config),\ldots,n_{\numcolors}(\config)$), the expected number of vertices with color $i$ in a region $V$ with boundary condition $j$ will be equal to the derivative of the log of the partition function $\ZFbd{j}(V) = \ZFbd{j}(\beta, \hvec, V)$ with respect to $h_i$:
\begin{align*}
\left\langle n_{i} \right\rangle_{\hvec,\beta} = \frac{1}{\ZFbd{j}(V)} \frac{\deriv \ZFbd{j}(V)}{\deriv h_i} = \frac{\deriv}{\deriv h_i} \log \ZFbd{j}(V).
\end{align*}
To construct $\congood$ for each $L$, we want to think about expected densities in regions of ``mostly $j$''.
However, the standard partition function for the grand canonical ensemble allows for large contours, and so even if we had set a boundary condition of $j$ in a region, the primary contributors to the partition function may not come from configurations where the color is ``mostly $j$''.

The precise definition of $\theta_{i,j}(\beta, \hvec)$ requires the partition function $Z(\beta, \hvec, \regionL)$ over a region $\regionL$ to be replaced by the truncated partition function $\ZFtr{j}(\beta, \hvec, \regionL)$ for a region of mostly $j$. Intuitively, the truncated partition function does not allow large flips in color, but for colors $j$ with the largest pressures, will roughly match the full partition function.
\begin{definition}[Expected Density]
\label{defn:expecteddensity}
Fix a magnetic field $\hvec \in \hoptopenset$. We denote by $\theta_{i,j}(\beta, \hvec)$ the expected density of color $i$ in a ``region of mostly $j$'', which we define as follows:
\[
\theta_{i,j}(\beta, \hvec) = \frac{\deriv}{\deriv h_i} \lim_{L \to \infty} \left(\frac{1}{|V(\Lbox)|} \log \ZFtr{j}(\beta, \hvec, V(\Lbox))\right).
\]
\end{definition}

We wish for the truncated pressures to be balanced, so for the rest of the discussion, we will be using $\hvecopt$ from \cref{lem:balance} as our magnetic field. For brevity, we will denote $\theta_{i,j} = \theta_{i,j}(\beta, \hvecopt)$ for the rest of this section. In addition, whenever a partition function is used without specifying the inverse temperature or magnetic field, the values are assumed to be $\beta$ and $\hvecopt$ respectively.

For each possible (large) value of $L$, we define the ``optimal contours set'' as $\congood = \congood(\beta, \hvecopt)$.
The contour set $\congood$ specifies the color of each interior, which partitions the region $V(\regionL) \setminus \overline{\congood}$ into $\numcolors$ (not necessarily connected) regions, denoted $\reggood{L}{1}$, $\reggood{L}{2}$, $\ldots$, $\reggood{L}{\numcolors}$, representing the unions of the components of $V(\regionL) \setminus \overline{\congood}$ that correspond to each color.
The following statements are written with $\numcolors = 3$ as an illustrative example to keep the matrices easy to read.
To construct $\congood$, we wish to have:
\begin{align}
\label{eqn:simeq}
\begin{split}
|\reggood{L}{1}|\theta_{1,1} + |\reggood{L}{2}|\theta_{1,2} + |\reggood{L}{3}|\theta_{1,3} &\approx \rho_1 \numparticles \\
|\reggood{L}{1}|\theta_{2,1} + |\reggood{L}{2}|\theta_{2,2} + |\reggood{L}{3}|\theta_{2,3} &\approx \rho_2 \numparticles \\
|\reggood{L}{1}|\theta_{3,1} + |\reggood{L}{2}|\theta_{3,2} + |\reggood{L}{3}|\theta_{3,3} &\approx \rho_3 \numparticles
\end{split}
\end{align}
Where the approximate equals sign is used due to rounding.
Thus by defining the following:
$$
\Theta = \begin{pmatrix}
\theta_{1,1} & \theta_{1,2} & \theta_{1,3} \\
\theta_{2,1} & \theta_{2,2} & \theta_{2,3} \\
\theta_{3,1} & \theta_{3,2} & \theta_{3,3}
\end{pmatrix}
\text{, }
\rhovec = 
\begin{pmatrix}
\rho_1 \\
\rho_2 \\
\rho_3
\end{pmatrix},
$$
We thus require that our target density $\rhovec$ is a within the convex hull of the column vectors of $\Theta$. 
As $\rhovec$ is a convex combination of the columns of $\Theta$, the following solution exists for $|\reggood{L}{1}|$, $|\reggood{L}{2}|$, $|\reggood{L}{3}|$ with entries in $\{0,1,2,\ldots,\numparticles\}$:
\[
\begingroup
\renewcommand*{\arraystretch}{1.25}
\begin{pmatrix}
|\reggood{L}{1}| \\
|\reggood{L}{2}| \\
|\reggood{L}{3}|
\end{pmatrix}
\endgroup
\approx \left(\Theta^{-1}\rhovec\right) \numparticles
\]
Thus, to construct the contour set $\congood$ itself (and thus the regions $\reggood{L}{1},\reggood{L}{2},\ldots,\reggood{L}{\numcolors}$), we first define the \emph{adjusted densities} $\rhovecbeta = \Theta^{-1}\rhovec$, and construct a minimal cost subdivision $\optconfign{\rhovecbeta}$ according to \cref{dfn:minimalcostsubdivision}. The validity of this construction is shown later in \cref{lem:rhovecvalidity}, as it requires additional properties of $\theta_{i,j}$ which are discussed and shown in \cref{sec:ceproperties}.
\begin{definition}[Minimal cost subdivision]
\label{dfn:minimalcostsubdivision}
Let $\rhovecp = (\rhop{j})_{j \in \colorset}$ be a vector of values in $[0,1]$ such that $\|\rhovecp\|_1 = 1$. 
Defining $\nvecp = \left(\floor{\rhop{1}\numparticles},\floor{\rhop{2}\numparticles},\ldots,\floor{\rhop{\numcolors-1}\numparticles},\numparticles - \sum_{i=1}^{\numcolors-1}\floor{\rhop{i}\numparticles}\right)$,
we say that a configuration $\config$ is a a minimum cost subdivision of $\regionL$ corresponding to $\rhovecp$ if it is a configuration in $\CFGPf{L}{\nvecp}$ satisfying $\Hcostbar{\config} \leq \Hcostbar{\config'}$ for all $\config' \in \CFGPf{L}{\nvecp}$. We commonly use $\optconfign{\rhovecp}$ to denote a minimal cost subdivision of $\regionL$.
\end{definition}
We thus define $\congood$ to be the set of contours of $\optconfign{\rhovecbeta}$. This also defines the regions $\reggood{L}{1}$ to $\reggood{L}{\numcolors}$.
We note however that even though $\optconfign{\rhovecbeta}$ is not the minimum cost set of contours for the original target densities $\rhovec$,
we will see later in \cref{lem:approximatebestcontour} that this can arbitrarily closely approximate the minimum cost subdivision corresponding to $\rhovec$, using sufficiently large values of $\beta$.

\subsection{Cluster expansions}
\label{sec:ceproperties}
In this section we will establish some properties of cluster expansions, which are relevant to our contour model. These properties will be used repeatedly throughout our proofs.
Throughout this paper, the constant $\minlength$ will denote the minimum possible value of $|\contourbar|$ for any contour $\contour$. On the triangular lattice, we have $\minlength = 7$.

\begin{lemma}[Bound on contour count]
\label{lem:pointcontourcount}
Denoting $\allcontours'$ to be $\allcontoursbd{j}$ or $\allcontoursperbdv{L}{j}{V}$ for some positive integer $L$, color $j \in \colorset$, and vertex set $V \subseteq V(\regionL)$, we have the following bound on the number of polymers of a fixed size $\ell \geq \minlength$ passing through a fixed point $x$:
\begin{align*}
\left|\left\{\contour \in \allcontours' : \contourbar \ni x, |\contourbar| = \ell \right\}\right| \leq \contourconst^\ell
\end{align*}
Where $\contourconst$ is some positive constant dependent only on the number of colors $\numcolors$.
\end{lemma}

\begin{proof}
The set $\contourbar$ is connected on the underlying graph, which has a degree bounded by $6$, so a simple upper bound for the number of possible vertex sets $\contourbar$ is $6^{2|\contourbar|}$ by a depth-first search. The number of possible values of $\contourassg{\contourbar}$ has a loose upper bound of $\numcolors^{\contourbar}$. This proves the lemma for $\contourconst = 36\numcolors$, although we expect this to be true for a significantly smaller value of~$\contourconst$.
\end{proof}

\begin{definition}[Ursell Function]
For a cluster $\cluster$, the Ursell function $\ursell$ is defined as follows:
\[
\ursell(\cluster) = \frac{1}{|\cluster|!}\sum_{\substack{H \subseteq G(\cluster)\\H \text{ connected,}\\\text{spannning}}}(-1)^{|E(H)|}
\]
\end{definition}

\begin{lemma}[Convergence Criterion for Cluster Expansions]
\label{lem:convergencecriterion}
Fix a color $j \in \colorset$ and
suppose that $w' : \allcontoursbd{j} \to \R_{\geq 0}$ is a $c\taub$-stable translation-invariant weight function assigning nonnegative real-valued weights to contours, where $c$ is independent of $\beta$ and $L$.
Let $\allcontours'$ be a finite set of contours in $\allcontoursbd{j}$ or $\allcontoursperbdv{L}{j}{V}$ for some positive integer $L$, color $j \in \colorset$, and vertex set $V \subseteq V(\regionL)$.
Then denoting by $\clusterset'$ the set of clusters corresponding to $\allcontours'$ and using $\contourconst$ from \cref{lem:pointcontourcount}, there exists a $\beta' = \beta'(\contourconst,c)$ such that for all $\beta \geq \beta'$, the cluster expansion 
\begin{align*}
\sum_{\cluster \in \clusterset'} \clw'(\cluster) \quad \text{where} \quad \clw'(\cluster) = \ursell(\cluster)\prod_{\contour \in \cluster}w'(\contour)
\end{align*}
converges. Furthermore, for any fixed vertex $x$ of $\fulllattice$ or $\regionL$ respectively and $\ell' \geq \minlength$, we have
\begin{align*}
\sum_{\substack{\cluster \in \clusterset', \clusterbar \ni x, |\clusterbar| \geq \ell'}} \left|\clw'(\cluster)\right| \leq e^{-c\taub(\ell' + \minlength)/4}.
\end{align*}
\end{lemma}

\begin{proof}
This proof relies on the application of the Koteck\'y--Preiss condition. However, we do not apply this to $w'(\contour)$, but to a new weight we will define as:
\begin{align*}
w''(\contour) = w'(\contour)e^{c\taub|\contourbar|/4} \leq e^{-3c\taub|\contourbar|/4}.
\end{align*}
Setting $a(\contour) = \frac{c\taub}{4}|\contourbar|$ for any contour $\contour$ and using $\contourconst$ from \cref{lem:pointcontourcount}, there exists a sufficiently large value of $\beta' = \beta'(\contourconst, c)$ such that for all $\beta \geq \beta'$, and any $\contour' \in \allcontours'$, we have
\begin{align*}
\sum_{\contour \in \allcontours', \contour \not\sim \contour'}|w''(\gamma)|e^{a(\contour)}
\leq |\contourbarp|\sum_{\ell \geq \minlength} \contourconst^\ell e^{-3c\taub \ell/4} e^{c\taub\ell/4} 
= |\contourbarp|\cdot \frac{1}{1-e^{\log \contourconst - c\taub/2}} e^{\minlength \left(\log \contourconst - c\taub/2\right)}
\leq \frac{c\taub}{4} |\contourbarp| = a(\contour).
\end{align*}
Then defining $\ds \clw''(\cluster) = \ursell(\cluster)\prod_{\contour \in \cluster}w''(\contour)$ and applying the Koteck\'y--Preiss condition (Theorem 5.4 of~\cite{Friedli2017}), we have absolute convergence of the sum over clusters with $\clw''$, which implies absolute convergence of the sum over clusters with $\clw'$, giving the first part of the lemma. In addition, for any fixed vertex $x$ of the underlying lattice we have the following upper bound for the sum over clusters containing $x$:
\begin{align*}
\sum_{\cluster \in \clusterset', \clusterbar \ni x} \left|\clw''(\cluster)\right|
\leq \sum_{\contour \in \allcontours', \contourbar \ni x} w''(\contour) e^{a(\contour)}
\leq \sum_{\ell \geq \minlength} \contourconst^\ell e^{-3c\taub \ell/4} e^{c\taub\ell/4} 
\leq e^{-c\taub\minlength/4}.
\end{align*}
Now applying this bound do the sum over the original weights $w'$, we have for any $\ell' \geq 0$,
\begin{align*}
\sum_{\substack{\cluster \in \clusterset', \clusterbar \ni x \\ |\clusterbar| \geq \ell'}} \left|\clw'(\cluster)\right|
&= \sum_{\substack{\cluster \in \clusterset', \clusterbar \ni x \\ |\clusterbar| \geq \ell'}} \left| \ursell(\cluster)\left(\prod_{\contour \in \cluster}e^{-c\taub|\contourbar|/4}\right) \left(\prod_{\contour \in \cluster}w'(\contour)e^{c\taub|\contourbar|/4}\right) \right|\\
&\leq e^{-c\taub \ell'/4} \sum_{\substack{\cluster \in \clusterset', \clusterbar \ni x \\ |\clusterbar| \geq \ell'}} \left| \clw''(\cluster) \right|
\leq e^{-c\taub(\minlength+\ell')/4}. \qedhere
\end{align*}

\end{proof}

To analyze the partition functions at fixed magnetization, we approximate them by their counterparts in the grand canonical emsemble. To do so, we need to understand what the typical magnetizations (i.e., the number of particles of each color) are in the grand canonical ensemble. This is done by analyzing the derivatives of the log partition function. In particular, the first derivative gives us the expected magnetization, while the second derivative tells us about its variance. This is explained in further detail in the proof of \cref{lem:magvarbound}.

\begin{lemma}[Derivatives of Truncated Partition Functions~\cite{Borgs1990}]
\label{lem:truncatedderivatives}
Fix a color $j \in \colorset$ and denote by~$\minlength$ the minimum possible value of $|\contourbar|$ for any contour $\contour$, and suppose that $\beta \geq \betamin$.
Let $\allcontours'$ be a finite set of contours in $\allcontoursbd{j}$ or $\allcontoursperbdv{L}{j}{V}$ for some positive integer $L$, color $j \in \colorset$, and vertex set $V \subseteq V(\regionL)$, and denoting by $\clusterset'$ the set of clusters corresponding to $\allcontours'$.
Then for all colors $i \in \colorset$ and $k \in \derivset$, the following series converges absolutely:
\begin{align*}
\derivki{k}{i} \sum_{\cluster \in \clusterset'} \clwtr{j}(\cluster) &= \sum_{\cluster \in \clusterset'} \dclwtr{j}(\cluster), \quad \text{where} \quad 
\clwtr{j}(\cluster) = \ursell(\cluster) \prod_{\contour \in \cluster} \wtr{j}(\contour).
\end{align*}
Notably, this means that the cluster expansion for $\derivki{k}{i} \log \left(e^{-h_j|V|}\ZPStr{L}{j}(\beta, \hvec, V)\right)$ converges. %

\noindent
Furthermore, for any vertex $x$ on the lattice $\fulllattice$ or $\regionL$ and $\ell' \geq \minlength$,, we have
\begin{align*}
\sum_{\cluster \in \clusterset', \clusterbar \ni x, |\clusterbar| \geq \ell'} \left|\dclwtr{j}(\cluster)\right| \leq e^{-\taub (\ell'+\minlength)/8}.
\end{align*}
In particular, taking $\ell' = \minlength$, this implies that 
\begin{align*}
\sum_{\cluster \in \clusterset', \clusterbar \ni x} \left|\dclwtr{j}(\cluster)\right| \leq e^{-\taub \minlength/4}.
\end{align*}
\end{lemma}

\begin{proof}
If our contour weights are $\taub$-stable, then by Lemma 2.3 of~\cite{Borgs1990}, the truncated weights and its derivatives:
\begin{align} \label{eqn:derivative_weight_bound}
\derivki{k}{i} \wtr{j}(\contour) \leq e^{-\taub|\contourbar|}.
\end{align}
Now we define a new set of weights
\begin{align} \label{eqn:weight_k_defn}
\wtrb{j}(\contour)
= e^{2/e}\max \left\{\wtr{j}(\contour), \frac{\deriv}{\deriv h_i}\wtr{j}(\contour), \derivki{2}{i} \wtr{j}(\contour)\right\}.
\end{align}
We consider the following cluster expansion over the contour set $\allcontours'$ (with cluster set $\clusterset'$) defined in the statement of the lemma and utilizing the weights $\wtrb{j}$:
\begin{align*}
\sum_{\cluster \in \clusterset'}\clwtrb{j}(X) := \sum_{\cluster \in \clusterset'}\ursell(X)\prod_{\contour \in X} \wtrb{j}(\contour)
\end{align*}
We use \cref{lem:convergencecriterion} with $\beta \geq \betamin$ (Appendix~\ref{apx:allconstants}) to show that this cluster expansion converges. To do this, we show the following bound when $\taub$ is sufficiently large:
\begin{align*}
\wtrb{j}(\contour) \leq e^{-\taub|\contourbar|/2}.
\end{align*}
By applying~\cref{eqn:derivative_weight_bound} to~\cref{eqn:weight_k_defn}, we have,
\begin{align*}
\wtrb{j}(\contour)
&\leq e^{2/e}\cdot e^{-\taub|\contourbar|} \\
&\leq e^{-\taub|\contourbar|/2} \text{ for } \taub \geq \frac{4}{e|\contourbar|}.
\end{align*}
As $\minlength$ is the minimum possible value of $|\contourbar|$ for any contour $\contour$ and we only consider $k \leq 2$, the bound on $\wtrb{j}(\contour)$ applies as long as $\taub \geq \frac{4}{e\minlength}$, which is true for $\beta \geq \betamin$ as in Appendix~\ref{apx:allconstants}. This gives us absolute convergence of the sum $\sum_{\cluster \in \clusterset'} \clwtrb{j}(X)$ uniformly on $h$, as well as the following upper bound on the total weight of all possible clusters (of size at least $\ell' \geq \minlength$) passing through any single vertex $x$ on the lattice $\fulllattice$ or $\regionL$:
\begin{align*}
\sum_{\cluster \in \clusterset': \clusterbar \ni x, |\clusterbar| \geq \ell'} \left|\clwtrb{j}(X)\right| \leq e^{\taub(\ell' + \minlength)/8}.
\end{align*}
These bounds on $\clwtrb{j}$ are useful to us because it upper bounds the $k$-derivatives of $\clwtr{j}$ for $k \in \derivset$:
\begin{align*}
\left|\derivki{k}{i} \clwtr{j}(X)\right|
&= |\ursell(X)| \cdot \left|\derivki{k}{i} \prod_{\contour \in X}\wtr{j}(\contour)\right| \\
&\leq |\ursell(X)| \cdot \left||X|^2 \prod_{\contour \in X} \frac{\wtrb{j}(\contour)}{e^{2/e}}\right| \\
&\leq \left|\ursell(X) \prod_{\contour \in X}\wtrb{j}(\contour)\right|  = \left|\clwtrb{j}(X)\right|,
\end{align*}
where the first inequality comes about because there are at most $|X|^2$ terms when applying the Leibniz Rule, and the second inequality is because $|X|^2 \leq e^{2|X|/e}$ for all $|X| \geq 1$.
This gives us the statement of the lemma. Note that the derivative and sum in the lemma statement can be swapped as the bound converges uniformly for all $\hvec$ within the open set $\hoptopenset$ containing $\hvecopt$ described in \cref{lem:balance}.
\end{proof}

The following lemma allows us to describe the cluster expansions of partition functions with a volume term and a boundary term. This is crucial for comparing cluster expansions across different scales (different values of $L$), and we do this by comparing contours to their counterparts defined on the infinite lattice $\fulllattice$.
One additional thing shown in \cref{lem:boundaryerrorbound} is that a restriction to only contractible contours has a negligible impact on the analysis.

\newcommand{\wrapset}[1]{\mathrm{wrap}^{#1}}
\begin{lemma}[Boundary error bound]
Fix a positive integer $L$, color $j \in \colorset$, and vertex set $V \subseteq V(\regionL)$. 
Suppose that $\clw'$ is either $\clwbd{j}$ or $\dclwtr{j}$ for $k \in \derivset$.
Then for all $\beta \geq \betamin$ and $L \geq \Lmin(\beta)$, we have
\label{lem:boundaryerrorbound}
\begin{align*}
\left|\sum_{\substack{\cluster \in \clustersetperbdv{L}{j}{V} }} \clw'(\cluster)
- |V|\sum_{\substack{\cluster \in \clustersetbd{j}, \clusterbar \ni 0}} \frac{1}{|\clusterbar|} \clw'(\cluster) \right|
\leq |\partialext V|\cdot 11e^{-\taub \minlength/8}.
\end{align*}
\end{lemma}
Note that the second term in the difference converges rapidly to $0$ as $L \to \infty$, so the first term dominates. The value multiplied with $|\partialext V|$ in the first term notably converges to $0$ as $\taub \to \infty$.

\newcommand{\clustersetnowrap}{\clustersetbd{j}^{\times}}
\newcommand{\distL}{d_{\regionL}}
\begin{proof}
Throughout this proof, we will repeatedly use the result from \cref{lem:convergencecriterion} with $\beta \geq \betamin$ (Appendix~\ref{apx:allconstants}) and \cref{lem:truncatedderivatives} that for~$\clusterset'$ denoting either $\clustersetperbdv{L}{j}{V}$ or $\clustersetbd{j}$, we have
\begin{align}
\label{eqn:clusterboundresult}
\sum_{\substack{\cluster \in \clusterset', \clusterbar \ni x, |\clusterbar| \geq \ell'}} \left|\clw'(\cluster)\right| \leq e^{-\taub(\ell' + \minlength)/8}.
\end{align}
We start with the following bound, which comes from applying~\cref{eqn:clusterboundresult} with $\ell' = L$, which allows us to restrict most of our discussion to clusters that are not large enough to wrap around $\regionL$:
\begin{align}
\label{eqn:boundaryerrorbound1}
\left| \sum_{\substack{\cluster \in \clustersetperbdv{L}{j}{V} }} \clw'(\cluster) - 
\sum_{\substack{\cluster \in \clustersetperbdv{L}{j}{V} \\ |\clusterbar| < L}} \clw'(\cluster) \right|
\leq \sum_{x \in V} \sum_{\substack{\cluster \in \clustersetperbdv{L}{j}{V}, \clusterbar \ni x \\ |\clusterbar| \geq L}} \left| \clw'(\cluster) \right|
\leq |V| e^{-\taub (\minlength+L)/8}.
\end{align}

Denote by $\wrapset{L} : V(\fulllattice) \to V(\regionL)$ a function that maps vertices of $\fulllattice$ to $\regionL$ by taking their positions modulo $L$.
We denote by $\clustersetnowrap$ the set of clusters from $\clustersetbd{j}$ using only polymers $\contour$ with interiors that do not intersect $V^{\compl}$ (specifically, when $V^\compl \cap \intr{j'}\contour = \emptyset$ for all colors $j' \in \colorset$).
Any cluster from $\clustersetnowrap$ that intersects $V$ would correspond to some cluster from $\clustersetperbdv{L}{j}{V}$, as every contour in the cluster is contractible in $V$. Thus, each cluster in $\clustersetperbdv{L}{j}{V}$ will be seen exactly $|\clusterbar|$ times in the double sum in~\cref{eqn:periodictofull}, which gives the first equality below.
\begin{align*}
\sum_{\substack{\cluster \in \clustersetperbdv{L}{j}{V} \\ |\clusterbar| < L}} \clw'(\cluster) 
\numberthis{eqn:periodictofull}
&= \sum_{x \in V} \sum_{\substack{\cluster \in \clustersetnowrap, \clusterbar \ni x \\ \wrapset{L}(\clusterbar) \subseteq (V \setminus \partialint V) \\ |\clusterbar| < L}} \frac{1}{|\clusterbar|} \clw'(\cluster) \\
&= \sum_{x \in V} \sum_{\substack{\cluster \in \clustersetnowrap, \clusterbar \ni x \\ |\clusterbar| < L}} \frac{1}{|\clusterbar|} \clw'(\cluster)
- \sum_{x \in V} \sum_{\substack{\cluster \in \clustersetnowrap, \clusterbar \ni x \\ \wrapset{L}(\clusterbar) \cap \partialint V \neq \emptyset \\ |\clusterbar| < L}} \frac{1}{|\clusterbar|} \clw'(\cluster) \\
&= \sum_{x \in V} \sum_{\substack{\cluster \in \clustersetnowrap, \clusterbar \ni x \\ |\clusterbar| < L}} \frac{1}{|\clusterbar|} \clw'(\cluster)
- \sum_{\substack{x \in \partialint V + (t_1L,t_2L), \\ t_1, t_2 \in \{-1,0,1\}}} \sum_{\substack{\cluster \in \clustersetnowrap, \clusterbar \ni x \\ \clusterbar \cap (V \setminus \partialint V) \neq \emptyset \\ |\clusterbar| < L}} \frac{1}{|\clusterbar \cap \partialint V|} \clw'(\cluster),
\end{align*}
where the final equality comes about because the restriction of the cluster size $|\clusterbar|$ to less than~$L$ ensures that the unwrapped $\clusterbar$ can only reach at most nine copies of $\partialint V$ (each translated by integer multiples of $L$).
Which means that:
\begin{align}
\label{eqn:boundaryerrorbound2}
\begin{split}
\left| \sum_{\substack{\cluster \in \clustersetperbdv{L}{j}{V} \\ |\clusterbar| < L}} \clw'(\cluster) 
- \sum_{x \in V} \sum_{\substack{\cluster \in \clustersetnowrap, \clusterbar \ni x \\ |\clusterbar| < L}} \frac{1}{|\clusterbar|} \clw'(\cluster) \right|
&\leq 9\left|\partialint V\right| \sum_{\substack{\cluster \in \clustersetbd{j}, \clusterbar \ni 0 \\ |\clusterbar| < L}} \frac{1}{|\clusterbar \cap \partialint V|} \left|\clw'(\cluster)\right| \\
&\leq 9\left|\partialint V\right| e^{-\taub \minlength/8},
\end{split}
\end{align}
where the last inequality comes from applying~\cref{eqn:clusterboundresult} with $\ell' = \minlength$.

To show that the contribution of contours that are non-contractible in $V$ is small, for each $x \in V$, we observe that for a cluster $\cluster$ containing $x$ to include a contour that is non-contractible in $V$, we must have $|\contourbar| \geq \distL(x,V^{\compl})$, where $\distL$ represents the graph distance over $\regionL$. Making use of~\cref{eqn:clusterboundresult} with $\ell' = \distL(x,V^{\compl})$, we have:
\begin{align*}
\left| \sum_{\substack{\cluster \in \clustersetbd{j}, \clusterbar \ni x \\ |\clusterbar| < L}} \frac{1}{|\clusterbar|} \clw'(\cluster)
- \sum_{\substack{\cluster \in \clustersetnowrap, \clusterbar \ni x \\ |\clusterbar| < L}} \frac{1}{|\clusterbar|} \clw'(\cluster) \right|
\leq \left| \sum_{\substack{\cluster \in \clustersetbd{j}, \clusterbar \ni x \\ |\clusterbar| \geq \distL(x,V^{\compl})}} \frac{1}{|\clusterbar|} \clw'(\cluster) \right|
\leq e^{-\taub (\minlength+\distL(x,V^{\compl}))/8}.
\end{align*}
This means that summing over the vertices in $V$, we have
\begin{align}
\label{eqn:boundaryerrorbound3}
\begin{split}
\sum_{x \in V} \left| \sum_{\substack{\cluster \in \clustersetbd{j}, \clusterbar \ni x \\ |\clusterbar| < L}} \frac{1}{|\clusterbar|} \clw'(\cluster)
- \sum_{\substack{\cluster \in \clustersetnowrap, \clusterbar \ni x \\ |\clusterbar| < L}} \frac{1}{|\clusterbar|} \clw'(\cluster) \right|
&\leq \sum_{k \geq 1} \left|\{x \in V \mid \distL(x,V^{\compl}) = k\}\right| e^{-\taub (\minlength+k)/8} \\
&\leq |\partialext V|e^{-\taub \minlength/8} \cdot \sum_{k \geq 1} 6k e^{-\taub k/8} \\
&= |\partialext V|e^{-\taub \minlength/8} \frac{6e^{\taub/8}}{(e^{\taub/8}-1)^2} \\
&\leq |\partialext V|e^{-\taub \minlength/8},
\end{split}
\end{align}
where we used the upper bound $\left|\{x \in V \mid \distL(x,V^{\compl}) = k\}\right| \leq 6k|\partialext V|$. The last inequality is true because $\beta \geq \betamin$ (Appendix~\ref{apx:allconstants}).
Finally, applying~\cref{eqn:clusterboundresult} once more with $\ell' = L$, we have for any $x \in V$,
\begin{align}
\label{eqn:boundaryerrorbound4}
\left| \sum_{\substack{\cluster \in \clustersetbd{j}, \clusterbar \ni x \\ |\clusterbar| < L}} \frac{1}{|\clusterbar|} \clw'(\cluster)
- \sum_{\substack{\cluster \in \clustersetbd{j}, \clusterbar \ni x}} \frac{1}{|\clusterbar|} \clw'(\cluster) \right|
\leq \sum_{\substack{\cluster \in \clustersetbd{j}, \clusterbar \ni x \\ |\clusterbar| \geq L}} \left| \clw'(\cluster) \right|
\leq e^{-\taub8 (\minlength+L)/8}.
\end{align}
Combine equations~\ref{eqn:boundaryerrorbound1},~\ref{eqn:boundaryerrorbound2},~\ref{eqn:boundaryerrorbound3}, and~\ref{eqn:boundaryerrorbound4} with the triangle inequality, we get
\begin{align*}
\left|\sum_{\substack{\cluster \in \clustersetperbdv{L}{j}{V} }} \clw'(\cluster)
- |V|\sum_{\substack{\cluster \in \clustersetbd{j}, \clusterbar \ni 0}} \frac{1}{|\clusterbar|} \clw'(\cluster) \right|
\leq |\partialext V|\cdot 10e^{-\taub \minlength/8} + 2|V|e^{-\taub(\minlength+L)/8}.
\end{align*}
This gives the statement of the lemma as $|V| \leq \numparticles$ and $2\numparticles e^{-\taub (\minlength + L)/8} \leq e^{-\taub \minlength / 8}$ for all $L \geq \Lmin(\beta)$ (Appendix~\ref{apx:allconstants}).
\end{proof}

\subsection{Geometric Bounds}
In this section, we establish some basic properties related to the densities $\theta_{i,j}(\beta,\hvec)$ and contour lengths, that will be useful later inr the proof. We start by showing that for $\hvec \in \hoptopenset$, the densities $\theta_{i,j}(\beta,\hvec)$ given by \cref{defn:expecteddensity} can be made arbitrarily close to $1$ in the case where $i=j$ and close to $0$ otherwise, by using a sufficiently large value of $\beta$.

\begin{lemma}
\label{lem:densityexpression}
Suppose that $\beta \geq \betamin$, $L \geq \Lmin(\beta)$, and $\hvec \in \hoptopenset$.
For each $i \in \colorset$, we can write the density $\theta_{i,j}(\beta,\hvec)$ in the following form:
\[
\theta_{i,j}(\beta,\hvec) = \mathbbm{1}_{i=j} + \sum_{\substack{\cluster \in \clustersetbd{j}, \clusterbar \ni 0}} \frac{1}{|\clusterbar|} \frac{\deriv}{\deriv h_i}\clwtr{j}(\cluster).
\]
In particular, by \cref{lem:truncatedderivatives} this means that
\[
\left|\theta_{i,j}(\beta,\hvec) - \mathbbm{1}_{i=j} \right| \leq e^{-\taub \minlength/4}.
\]
\end{lemma}

\begin{proof}
For each natural number $L$, we first write the following expression as a cluster expansion:
\begin{align}
\label{eqn:densitylimit}
\frac{1}{|V(\Lbox)|} \frac{\deriv}{\deriv h_i} \log \ZFtr{j}(\beta, \hvec, V(\Lbox))
&= \frac{1}{|V(\Lbox)|} \frac{\deriv}{\deriv h_i} \left( h_j|V(\regionL)| + \sum_{\clusterset \in \clustersetbdv{j}{V(\regionL)}} \clwtr{j}(\cluster) \right),
\end{align}
and follow a similar argument to \cref{lem:boundaryerrorbound} to show using \cref{lem:truncatedderivatives} that:
\begin{align*}
&\left| \frac{1}{|V(\Lbox)|} \sum_{\clusterset \in \clustersetbdv{j}{V(\regionL)}} \frac{\deriv}{\deriv h_i} \clwtr{j}(\cluster)
- \sum_{\clusterset \in \clustersetbd{j}, \cluster \ni 0} \frac{1}{|\clusterbar|} \frac{\deriv}{\deriv h_i} \clwtr{j}(\cluster) \right| \\
&\leq \frac{1}{|V(\Lbox)|} \sum_{x \in \partialint V} \sum_{\substack{\cluster \in \clustersetbd{j}, \clusterbar \ni x \\ \clusterbar \cap (V\setminus \partialint V) \neq \emptyset}} \frac{1}{|\clusterbar \cap \partialint V|}\left|\frac{\deriv}{\deriv h_i} \clwtr{j}(\cluster)\right| 
\leq \frac{|\partialint V(\regionL)|}{|V(\Lbox)|} e^{-\taub \minlength/4}.
\end{align*}
The expression in \cref{eqn:densitylimit} thus converges uniformly on $\hvec$, so by the definition of $\theta_{i,j}(\beta,\hvec)$ we have
\begin{align*}
\theta_{i,j}(\beta,\hvec)
&= \frac{\deriv}{\deriv h_i} \lim_{L \to \infty} \frac{1}{|V(\Lbox)|} \log \ZFtr{j}(\beta, \hvec, V(\Lbox)) \\
&= \lim_{L \to \infty} \frac{1}{|V(\Lbox)|} \frac{\deriv}{\deriv h_i} \log \ZFtr{j}(\beta, \hvec, V(\Lbox))
= \mathbbm{1}_{i=j} + \sum_{\substack{\cluster \in \clustersetbd{j}, \clusterbar \ni 0}} \frac{1}{|\clusterbar|} \frac{\deriv}{\deriv h_i}\clwtr{j}(\cluster). &\qedhere
\end{align*}
\end{proof}

The next two lemmas concern the adjusted densities $\rhovecbeta$ that we defined in \cref{subsec:construct good contours} (and used to construct $\congood$), which we show in \cref{lem:rhovecvalidity} is a valid construction and in \cref{lem:goodcontourapproximation} that it is a small perturbation of the pre-defined density vector $\rhovec$ assuming it exists.

\begin{lemma}[Close approximation of $\congood$]
\label{lem:goodcontourapproximation}
Given that $\beta \geq \betamin$ and $L \geq \Lmin(\beta)$, and $\rhovecbeta$ is a vector in $\R^\numcolors$ such that $\Theta \rhovecbeta = \rhovec$. Then
the difference between $\rhovecbeta$ and $\rhovec$ satisfies the following bound:
\begin{align*}
\| \rhovecbeta - \rhovec \|_1 \leq 2e^{-\taub \minlength/4}.
\end{align*}
\end{lemma}

\begin{proof}
The difference can be bounded by the matrix $1$-norm of $\Theta-I$, which we can bound using \cref{lem:densityexpression} as follows:
\begin{align*}
\| \rhovecbeta - \rhovec \|_1
&= \| (I - \Theta)\rhovecbeta \|_1 \leq \|I-\Theta\|_1 \|\rhovecbeta\|_1 = \|\Theta - I\|_1 
= \max_{j \in \colorset} \sum_{i \in \colorset} |\theta_{i,j} - \mathbbm{1}_{i=j}| \\
&= 2 \max_{j \in \colorset} (1-\theta_{j,j}) 
\leq 2\sum_{\cluster \in \clustersetbd{j}, \clusterbar \ni 0} \frac{1}{|\clusterbar|} \frac{\deriv}{\deriv h_i} \clwtr{j}(\cluster)
\leq 2e^{-\taub \minlength/4}. \qedhere
\end{align*}
\end{proof}

\begin{lemma}[Validity of $\rhovecbeta$ and $\congood$]
\label{lem:rhovecvalidity}
Suppose that $\beta \geq \betamin$ and $L \geq \Lmin(\beta)$.
Then the matrix $\Theta = (\theta_{i,j})_{i,j \in \colorset}$ is a non-singular stochastic matrix and $\rhovecbeta = \Theta^{-1}\rho$ is well-defined with strictly positive entries that sum to $1$.
\end{lemma}

\begin{proof}
Recall that $\theta_{i,j} = \theta_{i,j}(\beta,\hvecopt)$ by definition, so we have
\[
\theta_{i,j} = \frac{\deriv}{\deriv h_i} \lim_{L \to \infty} \left(\frac{1}{|V(\Lbox)|} \log \ZFtr{j}(\beta, \hvecopt, V(\Lbox))\right).
\]
To show that $\sum_{i \in \colorset} \theta_{i,j} = 1$ (i.e., that $\Theta$ is a stochastic matrix), we make use of the argument from \cref{lem:densityexpression} to exchange the derivative and the limit, and show that for each natural number $L$, we have
\begin{align*}
\sum_{i \in \colorset} \frac{\deriv}{\deriv h_i} \log \ZFtr{j}(\beta, \hvecopt, V(\Lbox))
&= \frac{1}{\ZFbd{j}(\beta, \hvecopt, V(\Lbox))} \sum_{i \in \colorset} \frac{\deriv}{\deriv h_i} \sum_{\config \in \CFGFint{j}{\regionL}} e^{-\beta H (\sigma) + \h \cdot \mathbf{n} (\sigma)} \\
&= \sum_{\config \in \CFGFint{j}{\regionL}} \dtbnh(\config) \sum_{i \in \colorset} n_i (\sigma) = |V(\regionL)|.
\end{align*}
This makes use of the equivalence of the truncated and untruncated partition functions (\cref{lem:truncatedequal}) and the fact that the $h_i$-derivative of a log partition function is the expected number of vertices of color $i$ ($n_i(\config)$ denoes the number of vertices of color $i$ in $\config$).

To show that $\Theta^{-1}$ exists, we make use of \cref{lem:densityexpression}. As $\hvecopt \in \hoptopenset$, we have
\begin{align*}
\left|\theta_{i,j} - \mathbbm{1}_{i=j} \right| \leq e^{-\taub \minlength/4} \text{ for }i, j \in \colorset,
\end{align*}
and as $\beta \geq \betamin$ implies that $e^{-\taub \minlength/4} < \frac{1}{q}$, the matrix $\Theta = (\theta_{i,j})_{i,j \in \colorset}$ is strictly diagonally dominant and thus non-singular.

Applying \cref{lem:goodcontourapproximation} and using the fact that $\beta \geq \betamin$ implies that $2e^{-\taub \minlength/4} \leq \frac{1}{2}\rhomin$, which shows that $\rhovecbeta$ has strictly positive entries. Then, using the fact that $\Theta \rhovecbeta = \rhovec$, we have
\begin{align*}
\sum_{j \in \colorset} \rhovecbeta_j
= \sum_{j \in \colorset} \rhovecbeta_j \sum_{i \in \colorset} \theta_{i, j}
= \sum_{i \in \colorset} \sum_{j \in \colorset} \rhovecbeta_j \theta_{i, j}
= \sum_{i \in \colorset} \rhovec_i = 1. &\qedhere
\end{align*}
\end{proof}

The following lemmas establish useful bounds for $\Hcostbar{\contourset}$ and $|\contoursetbar|$, which we can treat as the weighted and unweighted versions of the contour costs of a set of mutually compatible contours $\contourset$.

\begin{lemma}[Weighted and unweighted contour costs]
\label{lem:weightedcostratio}
Let $\contourset$ be a set of mutually compatible contours. Then the cost $\Hcostbar{\contourset}$ and the size of the support $\contoursetbar$ of the contour set are related by constant factors:
\begin{align*}
\frac{\costmatrixmin}{2} |\contoursetbar| \leq \Hcostbar{\contourset} \leq 3\costmatrixmax |\contoursetbar|.
\end{align*}
\end{lemma}

\begin{proof}
For any contour $\contour$ in $\contourset$,
a vertex $v$ is included in $\contourbar$ only if it is incident to some bichromatic edge (an edge between vertices of different colors) between vertices in $\contourbar$. This gives a lower bound of $\frac{1}{2}|\contourbar|$ bichromatic edges within the contour $\contour$, so $\Hcost{\contour} \geq \frac{\costmatrixmin}{2}|\contourbar|$.
On the other hand, as the induced subgraph $\regionL(\contourbar)$ has maximum degree $6$, it can have at most $3|\contourbar|$ edges, giving the other side of the bound. The same bounds apply to mutually compatible contour sets $\contourset$ as contours in such sets do not share vertices or edges.
\end{proof}

\begin{lemma}[Boundary costs and lengths]
\label{lem:mincostconfig}
For a vector $\rhovecp \in [0,1]^{\numcolors}$ where $\|\rhovecp\|_1 = 1$, 
denote by $\optconfign{\rhovecp}$ the configuration of minimal Hamiltonian on $\regionL$ corresponding to $\rhovecp$ as defined in \cref{dfn:minimalcostsubdivision}.
Denote $\rhopmin = \min_{i \in \colorset} \rhop{i}$.
Then there exists positive constants $\bdconstlow$, $\bdconstupp$, $\Hconstlow$, $\Hconstupp$ depending only on $\numcolors$, $\costmatrixmin$, and $\costmatrixmax$ such that
\begin{align*}
\Hconstlow \sqrt{\rhopmin} \cdot L \leq \hamil(\optconfign{\rhovecp}) \leq \Hconstupp L
\text{ and }
\bdconstlow \sqrt{\rhopmin} \cdot L \leq \left|\contoursetbar(\optconfign{\rhovecp})\right| \leq \bdconstupp L.
\end{align*}
\end{lemma}

\begin{proof}
We only aim to show that such constants exist, in that uniform lower and upper bounds for both $\hamil(\optconfign{\rhovecp})$ and $|\contoursetbar(\optconfign{\rhovecp})|$ that are linear on $L$ exist. A precise characterization of these constants is outside of the scope of this paper.

We define
\[
\nvecp = \left(\floor{\rhop{1}\numparticles},\floor{\rhop{2}\numparticles},\ldots,\floor{\rhop{\numcolors-1}\numparticles},\numparticles - \sum_{i=1}^{\numcolors-1}\floor{\rhop{i}\numparticles}\right)
\]
just like in the statement of \cref{dfn:minimalcostsubdivision}.
To show the upper bounds,
we define a configuration $\config' \in \CFGPf{L}{\nvecp}$ in the following manner:
Label the vertices of $\regionL$ row by row, from left to right, bottom to top.
We set the first $\npj{1}$ vertices in this ordering to color $1$, the next $\npj{2}$ vertices to color $2$, and so on. This configuration will have no more than $8L + 2L\numcolors$ bichromatic edges (edges between vertices of different colors). Thus, making use of \cref{lem:weightedcostratio}, we have the following upper bounds:
\begin{align*}
\frac{1}{2}\costmatrixmin\left|\contoursetbar(\optconfign{\rhovec})\right|
\leq \hamil(\optconfign{\rhovec}) \leq \hamil(\config')
\leq 3\costmatrixmax |\contoursetbar(\config')|
\leq 3\costmatrixmax (8 + 2\numcolors)L.
\end{align*}

To show the lower bounds,
let $i^*$ denote a color with the minimal amount of particles in $\nvecp$ (i.e., $\nj{i} = \min_{i \in \colorset}\{\nj{i}\} \leq \frac{1}{2}$).
Take any configuration $\config \in \CFGPf{L}{\nvecp}$, and let $U$ denote the set of vertices of color~$i^*$. There are two possibilities for $U$, depending on whether $U$ contains a closed path that can be contracted to a point without increasing in length or crossing $U^\compl$. If $U$ contains such a path, then necessarily $|\contoursetbar(\optconfign{\rhovec})| \geq 4L$. Otherwise, $U$ can be seen as a region of the infinite lattice $\fulllattice$, which has its boundary minimized when $U$ is a regular hexagon \cite{veomett2017general}. This gives a lower bound for $|\contoursetbar(\optconfign{\rhovec})|$ on the order of $\sqrt{|U|} = \sqrt{\rhomin} \cdot L + O_L(1)$. Applying \cref{lem:weightedcostratio}, gives us a similar lower bound for $\hamil(\optconfign{\rhovec})$.
\end{proof}

\subsection{Concentration results for magnetization}
As we will be able to see later in the proof of \cref{lem:magvarbound}, the first and second derivatives of (untruncated) partition functions allow us to analyze the expectation and variance of the number of particles of color $i$ in a region with boundary condition $j$.

The following lemma lower bounds the fixed-magnetization partition function $\ZFf{\theta_{\cdot,j}|\reggood{L}{j}|}(\reggood{L}{j})$ by the full partition function $Z(\reggood{L}{j})$ on the region $\reggood{L}{j}$ in the case where we fix the magnetization to its expected value.
Note that this lemma applies for all $\hvec \in \hoptopenset$, but we only use the lemma for $\hvec = \hvecopt$.
\begin{lemma}[Concentration of Magnetization]
\label{lem:magvarbound}
Let $\magrvi{i}$ denote the random variable denoting the number of sites of color $i$ in the region $\reggood{L}{j}$ with boundary condition $j$ in the grand canonical ensemble conditioned on there being no contours that are non-contractible in $\reggood{L}{j}$.
Then for all $\beta \geq \betamin$, $L \geq \Lmin(\beta)$, and $\hvec \in \hoptopenset$, we have
\begin{align*}
\max\left\{\big\|\theta_{i,j}(\beta,\hvec)|\reggood{L}{j}| - \E[\magrvi{i}]\big\|_2, \sqrt{\, 2\sum_{i \in\colorset}\mathrm{Var}(\magrvi{i})} \right\} \leq \left(|\partialext \reggood{L}{j}| + \sqrt{|\reggood{L}{j}|}\right)\cdot 11 \numcolors e^{-\taub \minlength/16}.
\end{align*}
\end{lemma}
\newcommand{\clustersumbound}{e^{-\taub\minlength/10}}
\newcommand{\ZM}{\ZPSbd{L}{j}(\beta, \hvec,\reggood{L}{j})}
\newcommand{\ZMtr}{\ZPStr{L}{j}(\beta, \hvec,\reggood{L}{j})}
\newcommand{\ZMs}{\ZPSbd{L}{j}(\beta, \hvecopt,\reggood{L}{j})}
\begin{proof}
We first observe that by differentiating the partition functions as sums over configurations, we have at any magnetic field $\hvec \in \hoptopenset$ (defined in \cref{lem:balance}) and color $i \in \colorset$:
\begin{align*}
\E[\magrvi{i}] &= \frac{1}{\ZM} \frac{\deriv}{\deriv h_i} \ZM = \frac{\deriv}{\deriv h_i} \log \ZM \text{and}\\
\E[\magrvi{i}^2] &= \frac{1}{\ZM}\frac{\deriv^2}{\deriv h_i^2} \ZM = \derivki{2}{i}  \log \ZM + \left(\frac{\deriv}{\deriv h_i} \log \ZM\right)^2.
\end{align*}
These are in terms of the untruncated partition functions, but through applying \cref{lem:truncatedequal}, we know that $\derivki{k}{i} \ZM = \derivki{k}{i} \ZMtr$ for $\hvec \in \hoptopenset$, $k \in \{1,2\}$.
By observing that
\begin{align*}
\E[\magrvi{i}]
= \frac{\deriv}{\deriv h_i} \log \ZMtr
= \mathbbm{1}_{i=j}|\reggood{L}{j}| + \sum_{\substack{\cluster \in \clustersetperbdv{L}{j}{V} }} \frac{\deriv}{\deriv h_i} \clwtr{j}(\cluster),
\end{align*}
we have by \cref{lem:boundaryerrorbound} and the expression for $\theta_{i,j}(\beta,\hvec)$ from \cref{lem:densityexpression}:
\begin{align*}
\left|\E[\magrvi{i}] - \theta_{i,j}(\beta,\hvec)|\reggood{L}{j}| \right|
\leq |\partialext \reggood{L}{j}|\cdot 11e^{-\taub \minlength/8}.
\end{align*}
Taking the $2$-norm over the vector $\left(\E[\magrvi{i}] - \theta_{i,j}(\beta,\hvec)|\reggood{L}{j}|\right)_{i \in \colorset}$, we have
\begin{align*}
\big\|\E[\magrvi{\cdot}] - \theta_{\cdot,j}(\beta,\hvec)|\reggood{L}{j}| \big\|_2 \leq |\partialext \reggood{L}{j}|\cdot 11 \sqrt{\numcolors}e^{-\taub \minlength/8}.
\end{align*}
Next, we bound the variance of $\magrv$, which we express as:
\begin{align*}
\mathrm{Var}(\magrvi{i})
= \E[\magrvi{i}^2] - \E[\magrvi{i}]^2 
= \derivki{2}{i} \log \ZMtr
= \sum_{\substack{\cluster \in \clustersetperbdv{L}{j}{V} }} \derivki{2}{i} \clwtr{j}(\cluster),
\end{align*}
where the last equality comes about because the term $h_j|\reggood{L}{j}|$ that we would expect from the definition of $\ZM$ is eliminated by the second order derivative, and the derivative and sum can be swapped due to uniform convergence on $\hoptopenset$.
Applying Lemmas~\ref{lem:boundaryerrorbound} and~\ref{lem:truncatedderivatives} and using the triangle inequality gives us for any $i \in \colorset$:
\begin{align*}
\sum_{\substack{\cluster \in \clustersetperbdv{L}{j}{\reggood{L}{j}} }} \left|\derivki{2}{i} \clwtr{j}(\cluster) \right|
&\leq |\reggood{L}{j}|\sum_{\cluster \in \clustersetbd{j}, \clusterbar \ni 0} \left|\derivki{2}{i} \clwtr{j}(\cluster) \right| + |\partialext \reggood{L}{j}|\cdot 11 \sqrt{\numcolors}e^{-\taub \minlength/8} \\
&\leq |\reggood{L}{j}|e^{-\taub \minlength/8} + |\partialext \reggood{L}{j}|\cdot 11 \sqrt{\numcolors}e^{-\taub \minlength/8} 
\end{align*}
which means that
\begin{align*}
\sqrt{\, 2\sum_{i \in\colorset}\mathrm{Var}(\magrvi{i})}
&\leq \sqrt{2\numcolors\left(|\reggood{L}{j}|e^{-\taub \minlength/8} + |\partialext \reggood{L}{j}|\cdot 11 \sqrt{\numcolors}e^{-\taub \minlength/8}\right)} \\
&\leq \sqrt{2\numcolors} \sqrt{|\reggood{L}{j}|}e^{-\taub \minlength/16} + |\partialext \reggood{L}{j}|\cdot 11 \numcolors e^{-\taub \minlength/16}, \\
\end{align*}
which gives us the statement of the lemma. Note that many of these bounds are intentionally loose for simplicity of writing. The dependency on $\numcolors$ for example can be optimized, but is outside the scope of this paper.
\end{proof}

There are three ingredients to being able to lower bound the fixed-magnetization partition function by that of the grand canonical ensemble (\cref{lem:regionfixedmag})---showing that $\theta_{\cdot,j}|V|$ is close to the expected magnetization in $|V|$ with boundary condition $j$, showing that the magnetization in the grand canonical ensemble concentrates, and finally showing some sort of ``Lipschitz'' condition, that small variations in the magnetization produce small variations in the fixed-magnetization log partition function. \cref{lem:magvarbound} covers the first two ingredients, and the next lemma (\cref{lem:smallchanges}) covers the third.

\begin{lemma}[Small changes in magnetization]
\label{lem:smallchanges}
Let $V$ be an arbitrary region in $\regionL$ and let $\mvec$, $\mvec'$ be two possible magnetizations such that $\|\mvec - \mvec'\|_2 \leq t$ for some $t \geq 0$. Then for $\beta \geq \betamin$, $L \geq \Lmin(\beta)$, and magnetic field $\hvecopt$, we have
\begin{align*}
\left|\log\ZPbdf{L}{j}{\mvec}(V) - \log\ZPbdf{L}{j}{\mvec'}(V) \right|
\leq t \cdot \smallchangesbound
\end{align*}
\end{lemma}

\begin{proof}
This can be proven with a simple Peierls argument.
For any configuration $\config \in \CFGPintf{L}{j}{\mvec}{V}$, we can convert it to a configuration $\config' \in \CFGPintf{L}{j}{\mvec'}{V}$ by changing the colors of at most $\|\mvec - \mvec\|_1 \leq t\sqrt{\numcolors}$ particles. The difference in contributions by $\config$ and $\config'$ to their partition functions at magnetic field $\hvecopt = (\hoptj{i})_{i \in \colorset}$ is thus bounded by:
\begin{align*}
\Bigg|\Big(-\beta \hamil(\config) + \sum_{i \in \colorset}\hoptj{i} n_i(\config)\Big) - \Big(-\beta \hamil(\config') + \sum_{i \in \colorset}\hoptj{i} n_i(\config')\Big)\Bigg| \leq \left(6\beta\costmatrixmax + \|\hvecopt\|_{\infty}\right)t\sqrt{\numcolors}.
\end{align*}
We note that the mapping from $\CFGPintf{L}{j}{\mvec}{V}$ to $\CFGPintf{L}{j}{\mvec'}{V}$ is not one-to-one; in the worst case, $\numcolors^{t\sqrt{\numcolors}}$ configurations may be mapped to the same configuration in $\CFGPintf{L}{j}{\mvec'}{V}$. The Peierls argument thus gives us the following bound,
\begin{align*}
\left|\log\ZPbdf{L}{j}{\mvec}(\beta, \hvecopt, V) - \log\ZPbdf{L}{j}{\mvec'}(\beta, \hvecopt, V) \right|
\leq \left(6\beta\costmatrixmax + \|\hvecopt\|_{\infty} + \numcolors \right)t\sqrt{\numcolors},
\end{align*}
which we can simplify to the bound in the lemma by observing that $\|\hvecopt\|_{\infty} \leq \hoptbound$ by \cref{lem:balance} and $\hoptbound + \numcolors \leq \beta \costmatrixmax$ by Appendix~\ref{apx:allconstants}.
\end{proof}

\subsection{Bounding fixed-magnetization partition functions using concentration results}

The following lemma shows that for each of the regions $\reggood{L}{j}$ with boundary condition $j$, the truncated partition function for a magnetization fixed to $\theta_{\cdot,j}|\reggood{L}{j}|$ is bounded below by the grand canonical (non-truncated) partition function on that same region with the same boundary condition.
For the remainder of this proof, we will assume an inverse temperature $\beta$ and magnetic field $\hvecopt$, and drop these two terms from our partition function notation for readability.

\begin{lemma}[Fixed magnetization on region]
\label{lem:regionfixedmag}
Denoting the ``rounded expected density'' vector
\begin{align*}
\floor{\theta_{\cdot,j}|\reggood{L}{j}|} := \Big( \floor{\theta_{1,j}|\reggood{L}{j}|}, \floor{\theta_{2,j}|\reggood{L}{j}|}, \cdots, \floor{\theta_{\numcolors-1,j}|\reggood{L}{j}|}, |\reggood{L}{j}| - \sum_{i=1}^{\numcolors-1}\floor{\theta_{i,j}|\reggood{L}{j}|} \Big)
\end{align*}
for each $j \in \colorset$, and given $\beta \geq \betamin$ and $L \geq \Lmin(\beta)$, we have
\begin{align*}
\ZPSbdf{L}{j}{\floor{\theta_{\cdot,j}|\reggood{L}{j}|}}(\reggood{L}{j}) \geq \frac{1}{\poly(L)}\ZPSbd{L}{j}(\reggood{L}{j}) \exp\left\{-\left(|\partialext \reggood{L}{j}| + \sqrt{|\reggood{L}{j}|}\right) \beta e^{-\taub \minlength/20}\right\}.
\end{align*}
\end{lemma}

\begin{proof}
Letting $\magrv = (\magrvi{1},\magrvi{2},\ldots,\magrvi{\numcolors})$ be as defined in \cref{lem:magvarbound},
by Chebyshev's inequality, we have for each $t > 0$,
\begin{align*}
Pr\left[\left\|\magrv - \E[\magrv]\right\|_2 > t\right] \leq \frac{\sum_{i \in\colorset}\mathrm{Var}(\magrvi{i})}{t^2}.
\end{align*}
So, using $\meanball$ to denote the set $\left\{\mvec \in \mathbb{Z}^{\numcolors}_{\geq 0}: \left\|\mvec-\E[\magrv]\right\|_2 \leq t, \|\mvec\|_1 = |\reggood{L}{j}|\right\}$ of possible magnetizations with distance less than $t$ to $\E[\magrv]$, we can rewrite the above inequality as
\begin{align}
\label{eqn1}
\sum_{\mvec \in \meanball} \ZPSbdf{L}{j}{\mvec}(\reggood{L}{j})
\geq \ZPSbd{L}{j}(\reggood{L}{j})\left(1 - \frac{\sum_{i \in\colorset}\mathrm{Var}(\magrvi{i})}{t^2}\right).
\end{align}
Thus, as long as we set $t$ large enough to have $\floor{\theta_{\cdot,j}|\reggood{L}{j}|} \in \meanball$ (i.e., $\big\|\floor{\theta_{\cdot,j}|\reggood{L}{j}|} - \E[\magrvi{\cdot}]\big\|_2 \leq t$), we have:
\begin{align*}
\ZPSbdf{L}{j}{\floor{\theta_{\cdot,j}|\reggood{L}{j}|}}(\reggood{L}{j})
&\geq \min_{\mvec \in \meanball} \left\{\ZPSbdf{L}{j}{\mvec}(\reggood{L}{j})\right\} \\
&\geq \max_{\mvec \in \meanball} \left\{\ZPSbdf{L}{j}{\mvec}(\reggood{L}{j})\right\}e^{-t \cdot \smallchangesbound} \text{ by \cref{lem:smallchanges}}\\
&\geq \frac{e^{-t \cdot \smallchangesbound}}{|\meanball|}\sum_{\mvec \in \meanball} \ZPSbdf{L}{j}{\mvec}(\reggood{L}{j}) \\
&\geq \frac{e^{-t \cdot \smallchangesbound}}{|\meanball|} \ZPSbd{L}{j}(\reggood{L}{j})\left(1 - \frac{\sum_{i \in\colorset}\mathrm{Var}(\magrvi{i})}{t^2}\right) \text{ by~\cref{eqn1}}.
\end{align*}
By setting $t = \left(|\partialext \reggood{L}{j}| + \sqrt{|\reggood{L}{j}|}\right)\cdot 11 \numcolors e^{-\taub \minlength/16} \geq \max\left\{\big\|\theta_{i,j}|\reggood{L}{j}| - \E[\magrvi{i}]\big\|_2, \sqrt{2\sum_{i \in\colorset}\mathrm{Var}(\magrvi{i})}\right\}$ as seen in \cref{lem:magvarbound} and adding a small constant to account for the floor function, 
we have
\begin{align*}
\ZPSbdf{L}{j}{\floor{\theta_{\cdot,j}|\reggood{L}{j}|}}(\reggood{L}{j})
&\geq \frac{\ZPSbd{L}{j}(\reggood{L}{j})}{2|\meanball|} \exp\left\{-\left(|\partialext \reggood{L}{j}| + \sqrt{|\reggood{L}{j}|}\right)\cdot 77 \numcolors^{3/2} \costmatrixmax \beta e^{-\taub \minlength/16}\right\} \\
&\geq \frac{\ZPSbd{L}{j}(\reggood{L}{j})}{2|\meanball|} \exp\left\{-\left(|\partialext \reggood{L}{j}| + \sqrt{|\reggood{L}{j}|}\right) \beta e^{-\taub \minlength/20}\right\},
\end{align*}
where the last inequality is a simplification of the bound using Appendix~\ref{apx:allconstants} as $\beta \geq \betamin$.
\end{proof}

Because of how the regions $\reggood{L}{j}$, $j \in \colorset$ have been defined, \cref{lem:regionfixedmag} can be applied to lower bound the fixed-magnetization partition function $Z_{\congood,\nfixedvec}$ over all configurations containing the contour $\congood$.

\begin{lemma}[Fixed magnetization lower bound]
\label{lem:goodineq}
Given that $\beta \geq \betamin$ and $L \geq \Lmin(\beta)$, we have
\begin{align*}
\ZPcf{L}{\congood}{\nfixedvec}
\geq \frac{1}{\poly(L)}\ZPc{L}{\congood} \exp\left\{-|\congoodbar |\beta e^{-\taub \minlength/24}\right\}.
\end{align*}
\end{lemma}

\newcommand{\mij}[2]{m_{#1,#2}} %
\newcommand{\Ngoodi}[1]{n_{#1}(\congood)} %
\newcommand{\roundingerror}[2]{J_{#1,#2}} %
\newcommand{\apNi}[1]{\overline{N}_{#1}} %
\newcommand{\apnj}[1]{\overline{n}_{#1}} %
\newcommand{\apmij}[2]{\overline{m}_{#1,#2}} %
\begin{proof}
We start with the following expression for $\ZPc{L}{\congood}$ from \cref{eqn:wrapexpansion}:
\begin{align*}
\ZPc{L}{\congood} = e^{\sum_{\contour \in \congood} \contourcost{\contour}} \prod_{j \in \colorset} \ZPSbd{L}{j}(\reggood{L}{j}).
\end{align*}
Denoting by $\Ngoodi{i}$ the number of particles of color $i \in \colorset$ in $\congoodbar$,
we can express the above expression in the fixed magnetization setting by taking a sum over all possible combinations for $\mij{i}{j} \in \{0,1,2,\ldots,\numparticles\}$, for $i,j \in \colorset$, as follows:
\begin{align*}
\ZPcf{L}{\congood}{\nfixedvec}
= e^{\sum_{\contour \in \congood} \contourcost{\contour}}
\sum_{\substack{(\mij{i}{j})_{i,j \in \colorset}: \\ \forall i, \sum_{j}\mij{i}{j} + \Ngoodi{i} = \nj{j} \\ \forall j, \sum_{i} \mij{i}{j}
= |\reggood{L}{j}|}}
\prod_{j \in \colorset} \ZPSbdf{L}{j}{\mij{\cdot}{j}}(\reggood{L}{j}).
\end{align*}
We focus on one specific term of this sum, and use it as a lower bound for $\ZPcf{L}{\congood}{\nfixedvec}$.
We construct this term $\mij{i}{j}$, $i, j \in \colorset$ as follows, where we use the floor function to ensure that we only work with integers, and add a small value $\roundingerror{i}{j}$ to each term to account for rounding errors that come from applying the floor function.
\begin{align*}
\mij{i}{j} = \floor{\theta_{i,j}|\reggood{L}{j}| + \left(\theta_{i,j} - \mathbbm{1}_{i=j} \right)\left( \rhobeta{j}\numparticles - |\reggood{L}{j}| \right)} + \roundingerror{i}{j}
\end{align*}
To show that this selection of $\mij{i}{j}$ is a term in the above sum, we observe the following approximate equalities, where $\approx_\numcolors$ means that these values differ by at most
$\numcolors$.
\begin{alignat*}{3}
\mij{i}{j} &\approx_\numcolors \apmij{i}{j} &&:= \theta_{i,j}|\reggood{L}{j}| + \left(\theta_{i,j} - \mathbbm{1}_{i=j} \right)\left( \rhobeta{j}\numparticles - |\reggood{L}{j}| \right) &&~\text{ for all $i,j \in \colorset$,} \\
\nj{i} &\approx_\numcolors \apnj{i} &&:= \rho_{i} \numparticles &&~\text{ for all $i \in \colorset$, and } \\
\Ngoodi{i} &\approx_\numcolors \apNi{i} &&:= \rhobeta{i}\numparticles - |\reggood{L}{i}| = \sum_{j \in \colorset}\mathbbm{1}_{i=j}\left( \rhobeta{j}\numparticles - |\reggood{L}{j}| \right) &&~\text{ for all $i \in \colorset$.}
\end{alignat*}
These values satisfy the following equations:
\begin{alignat*}{2}
\sum_{j \in \colorset}\apmij{i}{j} + \apNi{i}
&= \sum_{j \in \colorset}\theta_{i,j} \rhobeta{j} \numparticles 
= \rho_{i} \numparticles = \apnj{i} &&\text{ for each $i \in \colorset$, and}\\
\sum_{i \in \colorset}\apmij{i}{j}
&= 1 \cdot |\reggood{L}{j}| + 0 \cdot \left( \rhobeta{j}\numparticles - |\reggood{L}{j}| \right)
= |\reggood{L}{j}| &&\text{ for each $j \in \colorset$.}
\end{alignat*}
These equations are only approximate due to the application of floor functions, but it can be easily shown that all of these equations can be simultaneously satisfied in the integral case by applying small corrections using the terms $\roundingerror{i}{j}$ as seen in the definitions of $\mij{i}{j}$, where $|\roundingerror{i}{j}| \leq \numcolors^2$ for all $i, j \in \colorset$.
Therefore, due to our choice of $\mij{i}{j}$, for each $j \in \colorset$, by applying Lemmas~\ref{lem:truncatedderivatives} and~\ref{lem:densityexpression} we have
\begin{align*}
\left\| \floor{\theta_{\cdot,j}|\reggood{L}{j}|} - \mij{\cdot}{j} \right\|_2 
&\leq \left\| \left(\theta_{\cdot,j} - \mathbbm{1}_{\cdot=j} \right)\left( \rhobeta{j}\numparticles - |\reggood{L}{j}| \right) \right\|_2 + \poly(\numcolors) \\
&\leq  \left| \rhobeta{j}\numparticles - |\reggood{L}{j}| \right | \sqrt{\numcolors} e^{-\taub \minlength/4} + \poly(\numcolors).
\end{align*}
We can then make use of Lemmas~\ref{lem:smallchanges} and~\ref{lem:regionfixedmag} to obtain the following bound on $\ZPSbdf{L}{j}{\mij{\cdot}{j}}(\reggood{L}{j})$:
\begin{align*}
\ZPSbdf{L}{j}{\mij{\cdot}{j}}(\reggood{L}{j})
&\geq \frac{1}{\poly(L)}\ZPSbd{L}{j}(\reggood{L}{j}) e^{-F_j},\\
\text{where }
F_j &:= \left(|\partialext \reggood{L}{j}| + \sqrt{|\reggood{L}{j}|}\right) \beta e^{-\taub \minlength/20} + \left|\rhobeta{j}L^2 - |\reggood{L}{j}|\right|\sqrt{\numcolors}e^{-\taub\minlength/4} \cdot \smallchangesbound.
\end{align*}
Finally, making use of the following facts,
\begin{align*}
\sum_{j \in \colorset} |\partialext \reggood{L}{j}| &\leq |\congoodbar|,\\
\sum_{j \in \colorset} \left|\rhobeta{j}L^2 - |\reggood{L}{j}|\right| &\approx_{\numcolors} |\congoodbar|\text{, and}\\
\sum_{j \in \colorset} \sqrt{|\reggood{L}{j}|} &\leq \sqrt{\numcolors |V(\regionL)|} = \sqrt{\numcolors}L,
\end{align*}
we can combine the bounds for each regions $\reggood{L}{j}$ to obtain a lower bound for $\ZPcf{L}{\congood}{\nfixedvec}$ as follows:
\begin{align*}
\ZPcf{L}{\congood}{\nfixedvec}
&\geq e^{\sum_{\contour \in \congood} \contourcost{\contour}} \prod_{j \in \colorset} \ZPSbdf{L}{j}{\mij{\cdot}{j}}(\reggood{L}{j}) 
\geq \frac{1}{\poly(L)} e^{\sum_{\contour \in \congood} \contourcost{\contour}} \prod_{j \in \colorset} \ZPSbd{L}{j}(\reggood{L}{j}) e^{-F_j} \\
&\geq \frac{1}{\poly(L)}\ZPc{L}{\congood} \exp\left\{-|\congoodbar|\cdot 2\beta e^{-\taub \minlength/20} - L \cdot \sqrt{\numcolors}\beta e^{-\taub \minlength/20} \right\}.
\end{align*}
It remains to simplify the bound to obtain the statement of the lemma.
Denote $\rhomin = \min_{i\in\colorset}\rho_i$ and $\rhobetamin = \min_{i\in\colorset}\rhobeta{i}$.
From \cref{lem:goodcontourapproximation}, we have $\rhobetamin \geq \rhomin - 2e^{-\taub \minlength/4}$. Referring to Appendix~\ref{apx:allconstants} with the assumption that $\beta \geq \betamin$ gives us $\rhobetamin \geq \frac{1}{2}\rhomin$, which means that by \cref{lem:mincostconfig},
\begin{align*}
|\congoodbar| \geq \bdconstlow L \sqrt{\rhobetamin} \geq \bdconstlow L \sqrt{\rhomin/2}.
\end{align*}
Applying Appendix~\ref{apx:allconstants} again with the assumption that $\beta \geq \betamin$ gives us the statement of the lemma.
\end{proof}

\subsection{Comparing partition functions between $\conbad$ and $\congood$}
The most complex parts of the proof were to show that the partition functions for the fixed magnetization case and the grand canonical ensemble are comparable. Now with this assumption, we are left with showing that in the grand canonical ensemble, the ratio of the partition functions corresponding to the sets $\conbad$ and $\congood$ of the ``large'' contours is dominated by the costs $\Hcostbar{\conbad}$ and $\Hcostbar{\congood}$ of the ``large'' contours themselves.

\begin{lemma}[Comparing Grand Canonical Partition Functions]
\label{lem:grandineq}
Suppose that $\alphaA > 1$. Then as long as $\beta \geq \max\{\betamin, \betaminA(\alphaA)\}$ and $L \geq \Lmin(\beta)$,
the following bound is true for all contour sets $\conbad$ such that $\Hcostbar{\conbad} \geq \alphaA \Hcostbar{\congood}$:
\begin{align*}
\ZPc{L}{\conbad} \leq \ZPc{L}{\congood} \exp\left\{-\frac{\beta}{2}\left(\frac{\alphaA-1}{\alphaA}\right) \Hcostbar{\conbad} \right\}.
\end{align*}
\end{lemma}

\begin{proof}
\newcommand{\Kgood}{K_1}
\newcommand{\Kbad}{K_2}
\newcommand{\kgood}{{k}}
\newcommand{\kbad}{{k'}}
\newcommand{\Vgood}[1]{U_{#1}}
\newcommand{\Vbad}[1]{V_{#1}}
\newcommand{\Vgoodk}{U_{\kgood}}
\newcommand{\Vbadk}{V_{\kbad}}
\newcommand{\jgood}{j}
\newcommand{\jbad}{j'}
\newcommand{\jgoodk}{\jgood_{\kgood}}
\newcommand{\jbadk}{\jbad_{\kbad}}
\newcommand{\sumkgood}{\sum_{\kgood=1}^{\Kgood}}
\newcommand{\sumkbad}{\sum_{\kbad=1}^{\Kbad}}
Denote by $\Vgood{1},\Vgood{2},\ldots,\Vgood{\kgood}$ the vertex sets of the connected components of $V(\regionL)\setminus \congoodbar$, and $\Vbad{1},\Vbad{2},\ldots,\Vbad{\kbad}$ the vertex sets of the connected components of $V(\regionL)\setminus \conbadbar$. We get the following expressions for $\ZPc{L}{\congood}(\beta, \hvecopt)$ and $\ZPc{L}{\conbad}(\beta, \hvecopt)$ from applying \cref{eqn:wrapexpansion}:
\begin{align*}
\log \ZPc{L}{\congood}(\beta, \hvecopt)
&= \sum_{\contour \in \congood} \contourcost{\contour} + \sumkgood \log \ZPSbd{L}{\jgoodk}(\beta, \hvecopt, \Vgoodk) \\
\log \ZPc{L}{\conbad}(\beta, \hvecopt)
&= \sum_{\contour \in \conbad} \contourcost{\contour} + \sum_{\kbad = 1}^{\Kbad} \log \ZPSbd{L}{\jbadk}(\beta, \hvecopt, \Vbadk)
\end{align*}
By applying \cref{lem:boundaryerrorbound} and observing that $\sumkgood|\partialext \Vgoodk| \leq |\congoodbar|$ and $\sumkbad|\partialext \Vbadk| \leq |\conbadbar|$ we find
\begin{align*}
\left|\sumkgood \log \ZPSbd{L}{\jgoodk}(\beta, \hvecopt, \Vgoodk)
- \sumkgood|\Vgoodk|\sum_{\substack{\cluster \in \clustersetbd{j}, \clusterbar \ni 0}} \frac{1}{|\clusterbar|} \clwbd{\jgoodk}(\cluster) \right|
\leq |\congoodbar| \cdot 11e^{-\taub \minlength/8}
\end{align*}
\begin{align*}
\left|\sum_{\kbad = 1}^{\Kbad} \log \ZPSbd{L}{\jbadk}(\beta, \hvecopt, \Vbadk)
- \sum_{\kbad = 1}^{\Kbad}|\Vbadk|\sum_{\substack{\cluster \in \clustersetbd{j}, \clusterbar \ni 0}} \frac{1}{|\clusterbar|} \clwbd{\jbadk}(\cluster) \right|
\leq |\conbadbar| \cdot 11e^{-\taub \minlength/8}
\end{align*}
We observe that $\left|\sumkgood|\Vgoodk| - \sumkbad|\Vbadk|\right|
= \left||\congoodbar| - |\conbadbar|\right|$ and for each pair $i,i' \in \colorset$ we have 
\[
\sum_{\substack{\cluster \in \clustersetbd{j}, \clusterbar \ni 0}} \frac{1}{|\clusterbar|} \clwbd{i}(\cluster) = \pressuretr{i}(\hvecopt) = \pressuretr{i'}(\hvecopt) = \sum_{\substack{\cluster \in \clustersetbd{j}, \clusterbar \ni 0}} \frac{1}{|\clusterbar|} \clwbd{i'}(\cluster),
\]
so by applying the triangle inequality, we get
\begin{align*}
\left|\sumkgood \log \ZPSbd{L}{\jgoodk}(\beta, \hvecopt, \Vgoodk)
- \sum_{\kbad = 1}^{\Kbad} \log \ZPSbd{L}{\jbadk}(\beta, \hvecopt, \Vbadk)\right| 
&\leq \left(|\congoodbar| + |\conbadbar|\right) \cdot 11e^{-\taub \minlength/8}
+ \left||\congoodbar| - |\conbadbar|\right|\pressuretr{1}(\hvecopt)\\
&\leq \left(|\congoodbar| + |\conbadbar|\right) \cdot 11e^{-\taub \minlength/8}
+ \left||\congoodbar| - |\conbadbar|\right|e^{-\taub\minlength/4}.
\end{align*}
Applying the inequality $\|\hvecopt\|_{\infty} \leq \hoptbound$ from \cref{lem:balance} and once again making use of Appendix~\ref{apx:allconstants} with $\beta \geq \betamin$ to simplify the constants, we lower bound the ratio between the partition functions as follows:
\begin{align*}
\log \ZPc{L}{\congood}(\beta, \hvecopt) - \log \ZPc{L}{\conbad}(\beta, \hvecopt)
&\geq \sum_{\contour \in \congood} \contourcost{\contour} - \sum_{\contour \in \conbad} \contourcost{\contour} - \left(|\conbadbar| + |\congoodbar|\right) \cdot 12e^{-\taub \minlength/8} \\
&\geq \beta\Hcostbar{\conbad} - \beta\Hcostbar{\congood} - 2\left(|\conbadbar| + |\congoodbar|\right) \|\hvecopt\|_{\infty} - \left(|\conbadbar| + |\congoodbar|\right) \cdot 12e^{-\taub \minlength/8} \\
&\geq \beta\left(\frac{\alpha-1}{\alpha}\right)\Hcostbar{\conbad} - \left(|\conbadbar| + |\congoodbar|\right) \cdot 13e^{-\taub \minlength/8}
\end{align*}
Finally, we apply \cref{lem:weightedcostratio}, we may express both terms in terms of $\Hcostbar{\conbad}$ as follows:
\begin{align*}
|\conbadbar| + |\congoodbar| \leq \frac{2}{\costmatrixmin}\left(\Hcostbar{\conbad} + \Hcostbar{\congoodbar}\right) \leq \frac{4}{\costmatrixmin}\Hcostbar{\conbad}.
\end{align*}
This lets us simplify the bound to that of the lemma as long as $\beta \geq \betaminA(\alphaA)$, which by Appendix~\ref{apx:allconstants}, gives us:
\begin{align*}
&\beta \geq \frac{4}{\costmatrixmin} \cdot 13 e^{-\taub \minlength/8} \cdot \frac{2\alphaA}{\alphaA-1}. \qedhere
\end{align*}
\end{proof}

This allows us to conclude this section with the main lemma we wanted to show---that just like in the grand canonical ensemble, ratio of the fixed-magnetization partition functions corresponding to the sets $\conbad$ and $\congood$ of the ``large'' contours is dominated by the costs $\Hcostbar{\conbad}$ and $\Hcostbar{\congood}$ of the ``large'' contours themselves.

\begin{lemma}[Comparing Fixed-Magnetization Partition Functions]
\label{lem:finalcontourbound}
Suppose that $\alphaA > 1$. Then as long as $\beta \geq \max\{\betamin, \betaminA(\alphaA)\}$ and $L \geq \Lmin(\beta)$,
the following bound is true for all contours $\conbad$ where $\Hcostbar{\conbad} \geq \alphaA \Hcostbar{\congood}$:
\begin{align*}
\ZPcf{L}{\conbad}{\nfixedvec} 
&\leq \poly(L) \cdot \ZPcf{L}{\congood}{\nfixedvec} \exp\left\{-\frac{\beta}{4} \left(\frac{\alphaA-1}{\alphaA}\right) \Hcostbar{\conbad} \right\}.
\end{align*}
Furthermore, as $\ZPcf{L}{\congood}{\nfixedvec}$ is upper bounded by the full fixed-magnetization partition function $\ZPf{L}{\nfixedvec}$, this gives an upper bound on the probability of drawing a configuration from $\CFGPcf{L}{\conbad}{\nfixedvec}$:
\begin{align*}
\dtbnf(\CFGPcf{L}{\conbad}{\nfixedvec}) \leq \poly(L) \cdot \exp\left\{-\frac{\beta}{4} \left(\frac{\alphaA-1}{\alphaA}\right) \Hcostbar{\conbad} \right\}.
\end{align*}
\end{lemma}

\begin{proof}
We show this final bound by combining the earlier lemmas as follows:
\begin{align*}
\ZPcf{L}{\conbad}{\nfixedvec} &\leq \ZPc{L}{\conbad} &\text{(by~\cref{eqn:badineq})}\\
&\leq \ZPc{L}{\congood} \exp\left\{-\frac{\beta}{2} \left(\frac{\alphaA-1}{\alphaA}\right) \Hcostbar{\conbad} \right\} &\text{(by \cref{lem:grandineq})}\\
&\leq \poly(L) \cdot \ZPcf{L}{\congood}{\nfixedvec} \exp\left\{|\congoodbar| \beta e^{-\taub \minlength/24}\right\}\exp\left\{-\frac{\beta}{2} \left(\frac{\alphaA-1}{\alphaA}\right) \Hcostbar{\conbad} \right\} &\text{(by \cref{lem:goodineq})}.
\end{align*}
Applying \cref{lem:weightedcostratio} allows us simplify the bound with the following inequality
\begin{align*}
\Hcostbar{\conbad} &\geq \alphaA \Hcostbar{\congood} \geq \alphaA \frac{\costmatrixmin}{2}|\congoodbar|,
\end{align*}
which gives us the statement of the lemma for $\beta \geq \betaminA(\alphaA)$, as by Appendix~\ref{apx:allconstants}, we would have:
\begin{align*}
&\alphaA \frac{\costmatrixmin}{2} \cdot \frac{\beta}{4}\left(\frac{\alphaA-1}{\alphaA}\right) \geq \beta e^{-\taub \minlength/24}. \qedhere
\end{align*}
\end{proof}

\subsection{Approximating optimal shapes}
All of our analysis thus far has focused on comparing the cost of configurations to $\congood$, which is not the minimal cost subdivision corresponding to $\rhovec$ but a perturbation $\rhovecbeta$ of it. This perturbation is small by \cref{lem:goodcontourapproximation}, and we wish to show that this corresponds to a small difference in the costs of the respective minimal cost subdivisions.

\begin{lemma}[Approximating the minimal cost subdivision]
\label{lem:approximatebestcontour}
Suppose that $\optconfign{\rhovec}$ and $\optconfign{\rhovecbeta}$ refer to the minimal cost subdivisions of $\regionL$ corresponding to $\rhovec$ and $\rhovecbeta$ respectively.
Then for any $\alphaB > 1$, as long as $\beta \geq \max\{\betamin,\betaminB(\alphaB)\}$ and $L \geq \Lmin(\beta)$, we have
\begin{align*}
\Hcostbar{\congood} = \hamil(\optconfign{\rhovecbeta}) \leq \alphaB \hamil(\optconfign{\rhovec}).
\end{align*}
\end{lemma}

\newcommand{\Hcostoptconfig}{\Hcostbar{\contourset(\optconfign{\rhovec})}}
\begin{proof}
Let $\nvecbeta = (\nbetaj{i})_{i \in \colorset} = \left(\floor{\rhobeta{1} L^2}, \floor{\rhobeta{2} L^2}, \ldots, \floor{\rhobeta{\numcolors-1} L^2}, L^2 - \sum_{i=1}^{\numcolors-1}\floor{\rhobeta{i} L^2}\right)$. It suffices to find a configuration $\config \in \CFGPf{L}{\nvecp}$ such that $\hamil(\config) \leq \alphaB \hamil(\optconfign{\rhovec})$.

Suppose that $\contourset(\optconfign{\rhovec})$ is the set of contours of the minimal cost subdivision $\optconfign{\rhovec}$. This means that $\Hcostoptconfig = \hamil(\optconfign{\rhovec})$. We first show that there exists a configuration $\config^*$ with a single contour $\contour^*$ such that $\Hcost{\contour^*} = \Hcostoptconfig$.
Let $\contourset^*$ be the smallest set of contours such that $\Hcostbar{\contourset^*} = \Hcostoptconfig$. If $|\contourset^*| \geq 2$, then take any contour $\contour \in \contourset^*$ and translate it one vertex at a time in the direction of another contour and stop when it first becomes incompatible with another contour $\contour'$ in $\contourset^*$.
Incompatibility first happens when $d(\contourbar, \contourbar') = 1$, putting $\contourbar$ and $\contourbar'$ in the same component.
Translating a contour does not change its cost $\Hcost{\contour}$, and even though moving a contour requires changing the colors of particles, the total number of particles of each color is conserved.
This gives us a configuration of equal cost and with fewer contours than $\contourset^*$, which is a contradiction, so we must have $|\contourset^*| = 1$, showing the existence of $\config^*$.

Using $\config^* \in \CFGPf{L}{\nfixedvec}$, we now attempt to construct a new configuration $\config \in \CFGPf{L}{\nvecbeta}$ with $\hamil(\config) \leq \alphaB \hamil(\optconfign{\rhovec})$.
To do so, we start by showing that for each color $i \in \colorset$, there is a sufficiently large square region containing only particles of color $i$, which we can use to flip colors to construct $\config$.

We know that $|\contourbar^*| \leq \frac{2}{\costmatrixmin} \hamil(\optconfign{\rhovec}) \leq \frac{2\Hconstupp}{\costmatrixmin} L$ by \cref{lem:mincostconfig}.
Now define $x \in (0,1)$ as follows:
\begin{align*}
x = \min\left\{\frac{1}{6}, \frac{\costmatrixmin}{48\Hconstupp}\right\} \rhomin,
\end{align*}
where $\rhomin = \min_{i\in\colorset}\rho_i$.
We show that for each $i \in \colorset$, we will be able to find a $\floor{xL} \times \floor{xL}$ square region of $V(\regionL)$ that contains only vertices of color $i$.
Divide the region $V(\regionL)$ into $\floor[\Big]{\frac{1}{\floor{xL}}}^2$ nonoverlapping squares of side length $\floor{xL}$ each.
We start by showing that $\contourbar^*$ overlaps at most $\max\left\{9, \frac{16\Hconstupp}{x\costmatrixmin}\right\}$ of these squares. Enumerate the square regions as $S_1, S_2, \ldots, S_M$, and for each of these square regions $S_m$, let $S^+_m$ denote a $3\floor{xL} \times 3\floor{xL}$ region centered at $S_m$. If $S_m \cap \contourbar^* \neq \emptyset$, then as long as at least $10$ squares are overlapped, $\contourbar^*$ must connect $S_m$ to $\partialext S^+_m$, which implies that $|\contourbar^* \cap (S^+_m \setminus S_m)| \geq \floor{xL}$.
Due to how these squares are constructed, the sum of $|\contourbar^* \cap (S^+_m \setminus S_m)|$ over all squares overlapped by $\contourbar^*$ will count each vertex of $\contourbar^*$ at most $8$ times, so we must have:
\begin{align*}
\left|\{m: S_m \cap \contourbar^* \neq \emptyset\}\right| \cdot \floor{xL} \leq \sum_{m: S_m \cap \contourbar^* \neq \emptyset} |\contourbar^* \cap (S^+_m \setminus S_m)| \leq 8|\contourbar^*| \leq \frac{16\Hconstupp}{\costmatrixmin}L.
\end{align*}
This means that at least $\floor[\big]{\frac{L}{\floor{xL}}}^2 - \frac{1}{\floor{xL}} \cdot \frac{16\Hconstupp}{\costmatrixmin}L$ squares must have no intersection with $\contourbar^*$. Consequently, all vertices within any one of these squares must share the same color. Fix a color $i \in \colorset$. There are at most $(1 - \rho_i)L^2 + \numcolors$ vertices of $V(\regionL)$ that are not of color $i$. Thus, if $\left(\floor[\big]{\frac{L}{\floor{xL}}}^2  - \frac{1}{\floor{xL}} \cdot \frac{16\Hconstupp}{\costmatrixmin}L \right) \cdot \floor{xL}^2 \geq (1-\rho_i)L^2 + \numcolors$, there must be at least one $\floor{xL} \times \floor{xL}$ square that contains only vertices of color $i$. As
\begin{align*}
\left(\floor[\Bigg]{\frac{L}{\floor{xL}}}^2  - \frac{1}{\floor{xL}} \cdot \frac{16\Hconstupp}{\costmatrixmin}L \right) \cdot \frac{\floor{xL}^2}{L^2}
&\geq \left(\frac{L}{\floor{xL}} - 1 \right)^2 \cdot \frac{\floor{xL}^2}{L^2} - \frac{\floor{xL}}{L} \frac{16\Hconstupp}{\costmatrixmin} \\
&\geq (1-x)^2 - x \frac{16\Hconstupp}{\costmatrixmin},
\end{align*}
It thus suffices to show that $(1-x)^2 - x \frac{16\Hconstupp}{\costmatrixmin} \geq (1-\rho_i) + \frac{\numcolors}{L^2}$.
From our choice of $x$ and because we have $L \geq \Lmin(\beta)$ as given in Appendix~\ref{apx:allconstants}, we have $x \leq \frac{\rho_i}{6}$, $x \leq \frac{\rho_i}{3}\cdot \frac{\costmatrixmin}{16\Hconstupp}$, and $\frac{\numcolors}{L^2} < \frac{\rho_i}{3}$, which gives us the above inequality, showing the existence of the square.

\newcommand{\colortrans}[2]{T_{#1 \to #2}}
\newcommand{\cumsump}[1]{C'_{#1}}
\newcommand{\cumsum}[1]{C_{#1}}
The next step is change the colors of the appropriate number of particles from $\config^*$ to obtain a configuration in $\CFGPf{L}{\nvecbeta}$. We first partition the colors $\colorset$ into two sets $J =\{j_1,j_2,\ldots,j_k\}$ and $J' = \{j'_1, j'_2, \ldots, j'_{k'}\}$, with the former referring to colors $j$ with $\nbetaj{j} \leq \nj{j}$, and the latter colors $j$ with $\nbetaj{j} < \nj{j}$.
We first determine how many particles of each color we wish to ``transfer'' from $J$, the colors in excess, to $J'$, the deficient colors. For each pair $i, j \in \colorset$, we denote by $\colortrans{i}{j}$ the number of particles of color $i$ we will convert to color $j$. Hence, we would like the following property to hold for each $j \in \colorset$:
\begin{align*}
\nj{j} + \sum_{i \in \colorset}\colortrans{i}{j} - \sum_{i' \in \colorset}\colortrans{j}{i'} = \nbetaj{j}.
\end{align*}
Defining the cumulative sums of the excesses and deficiencies respectively as $\cumsum{r} = \sum_{i=1}^r (\nj{j_i} - \nbetaj{j_i})$ for $r \in [k]$
and $\cumsump{r'} = \sum_{i=1}^{r'} (\nj{j'_i} - \nbetaj{j'_i})$ for $r' \in [k']$, we can define for $r \in [k]$, $r' \in [k']$:
\begin{align*}
\colortrans{j_r}{j'_{r'}} = \max\left\{0, \min\{\cumsum{r}, \cumsump{r'}\} - \max\{\cumsum{r-1}, \cumsump{r'-1}\}\right\}
\end{align*}
For all other pairs $i,j \in \colorset$ not covered by the above, we set $\colortrans{i}{j} = 0$. This is equivalent to greedily packing the excess colors from $J$ to $J'$, in increasing index order. A consequence of redistributing the colors in this order is that $\sum_{j \in J} |\{i \in \colorset: \colortrans{j}{i} > 0\}| \leq |J| + |J'| = \numcolors$.

We now explain we flip colors in $\config^*$ to create a new configuration $\config$.
For each color $j \in J$, there exists an $\floor{xL} \times \floor{xL}$ box containing only particles of color $j$. To flip the $N_j := \nj{j} - \nbetaj{j} \sum_{i \in \colorset}\colortrans{j}{i} \geq 0$ particles to the desired colors, we define a $\ceil[\big]{\sqrt{N_j}} \times \ceil[\big]{\sqrt{N_j}}$ square region within the box, and iterate through this square region row by row, flipping the first $\colortrans{j}{1}$ particles to color $1$, the next $\colortrans{j}{2}$ particles to color $2$, and so on.

From \cref{lem:goodcontourapproximation}, we have $N_j \leq \|\nvecbeta - \nfixedvec\|_1 \leq 2L^2 e^{-\taub \minlength/4}$. Because $\beta \geq \betamin$ and $L \geq \Lmin(\beta)$ as given in Appendix~\ref{apx:allconstants}, we have $\ceil[\big]{\sqrt{N_j}} \leq \sqrt{N_j}+1 \leq 2Le^{-\taub \minlength/2} \leq xL$. This implies that $\ceil[\big]{\sqrt{N_j}} \leq \floor{xL}$, so the square region fits within the box.

Flipping the colors in this square region creates at most $\ceil[\big]{\sqrt{N_j}}\left(8 + 2|\{i \in \colorset: \colortrans{j}{i} > 0\}| \right)$ new bichromatic edges in the configuration. Taking the sum over $j \in J$, we can bound the difference in Hamiltonian by:
\begin{align*}
\hamil(\config) - \hamil(\config^*)
&\leq \sum_{j \in J} \ceil[\big]{\sqrt{N_j}}\big(8 + 2|\{i \in \colorset: \colortrans{j}{i} > 0\}| \big) \\
&\leq (8 \numcolors + 2\numcolors) \cdot 2Le^{-\taub \minlength/2}.
\end{align*}
Recalling that $\hamil(\config^*) = \hamil(\optconfign{\rhovec})$ and $\hamil(\optconfign{\rhovec}) \geq \Hconstlow L \sqrt{\rhomin}$ from \cref{lem:mincostconfig}, we obtain the statement of the lemma as long as $\beta \geq \betaminB(\alphaB)$, which by Appendix~\ref{apx:allconstants} gives us:
\begin{align*}
&20 \numcolors e^{-\taub \minlength/2} \leq (\alphaB-1)\Hconstlow \sqrt{\rhomin}. \qedhere
\end{align*}
\end{proof}
So far, we have been showing that sets of ``large'' contours $\confinal$ with high costs $\Hcostbar{\confinal}$ have exponentially small probability, without showing exactly how these large contours are identified. We apply a technique from~\cite{Miracle2011} known as bridging to identify the important contours from the contour set of any configuration. More precisely, the process of bridging will always produce a set $\confinal$ of contours that meets the second property of our key event, that the components of $V(\regionL) \setminus \contoursetbar$ are each dominated by a single color, with a small $\delta$-fraction of errors.
\newcommand{\bdgfunc}{\mathcal{B}_{\delta}}
\newcommand{\bdgvert}{B}
\newcommand{\bdgcont}{\conbad}
\newcommand{\bdgcontbar}{\overline{\conbad}}
\newcommand{\bdgvertd}{B_\delta}
\newcommand{\bdgcontd}{\conbad_\delta}
\newcommand{\bdgcontbard}{\overline{\conbad}_\delta}
\newcommand{\column}{C}
\begin{lemma}[$\delta$-Bridge Systems]
\label{lem:bridgesystems}
For each configuration $\config \in \CFGPf{L}{\nfixedvec}$ we can construct a
\emph{bridge system} $(\bdgvert, \bdgcont)$
where $\bdgvert \subseteq V(\regionL)$ is a set of bridging vertices and $\bdgcont$ is a subset of the contours of $\config$ that we call the ``bridged'' contours, with the following properties:
\begin{enumerate}
\item $\bdgvert \cup \bdgcontbar$ is a connected set of vertices in $\regionL$.
\item $|\bdgvert| \leq \frac{1}{2\delta}|\bdgcontbar|$ + L
\item $\bdgcont$ is closed under contractibility in $\config$, which implies that (the vertex set of) each component $V_k$ of the subgraph induced by $V(\regionL) \setminus \confinalbar$ has a well-defined label $j_k$.
\item For each component $V_k$ with label $j_k$, there are at most $\delta|V_k|$ particles of colors other than $j_k$.
\end{enumerate}
\end{lemma}
\begin{proof}
This proof is based off the construction of bridge systems in~\cite{Cannon2019,Kedia2022}.
We start by showing the existence of a pair $(\bdgvert_0, \bdgcont_0)$ that satisfies the first three properties of the lemma.
We define $\bdgvert_0 \subseteq V(\regionL)$ to be the set of all vertices of a single row (we will use the bottom-most row) of $V(\regionL)$, and define $\bdgcont_0$ to be the set of all contours $\contour$ where $d(\contourbar, \bdgvert_0) \leq 1$. Property 1 is clear from the construction, and property 2 is true because $|\bdgvert_0| = L$. To show property 3, let $U$ be a component of $V(\regionL) \setminus \bdgcontbar_0$. If there is a contour of $\config$ that is non-contractible in $U$, then it must intersect $\bdgvert_0$, and thus must already be in $\bdgcont_0$, so $U$ would not be a single component of $V(\regionL) \setminus \bdgcontbar_0$.

Now suppose that $\bdgcont$ is the largest set of contours that satisfies the first three properties with some set $\bdgvert$ of bridging vertices. We show that this set satisfies the fourth property as well.

Suppose that it does not, which means that there exists a component $V_k$, with label $j_k$, containing more than $\delta |V_k|$ particles of colors other than $j_k$ in $\config$. 
Partition the region $V(\regionL)$ into its $L$ columns, and consider the intersection of each column with $V_k$. There must exist a column $\column \subseteq V_k$ that contains more than $\delta |\column|$ particles of colors other than $j_k$. Consider the set $\contourset$ of all contours $\contour$ of~$\config$ where $\contourbar \subseteq V_k$ and $d(\contourbar, \column) \leq 1$. As every particle of color other than $j_k$ in $\column$ must have some contour in $\contourset$ that surrounds it, we must have $|\contoursetbar| \geq 2\delta|\column|$. Thus, defining $\bdgvert' = \bdgvert \cup \column$ and $\bdgcont' = \bdgcont \cup \contourset$, we must have $|\bdgvert'| \leq \frac{1}{2\delta}|\bdgcont'| + L$. Properties 1 and 2 are thus satisfied by $\bdgvert'$ and $\bdgcont'$. Property 3 is true for the same reason as that of $(\bdgvert_0, \bdgcont_0)$---if $V(\regionL)\setminus \bdgcontbar$ had a component $U$ containing a contour~$\contour$ that is non-contractible in it, it would have been included in $\contourset$ as it would intersect the column $\column$. Thus, we would have constructed a strictly larger set $\bdgcont'$ that satisfies the first three properties, contradicting the maximality of $\bdgcont$, thus showing that must satisfy the fourth property as well.
\end{proof}
For each $\delta > 0$, we denote by $\bdgfunc$ a function that maps each configuration $\config \in \CFGPf{L}{\nfixedvec}$ to a $\delta$-bridge system $\bdgfunc(\config) = (\bdgvertd(\config), \bdgcontd(\config))$ as given by \cref{lem:bridgesystems}.
The main strategy to show that configurations $\config$ with bridge systems with high cost $\Hcostbar{\bdgcontd(\config)}$ are unlikely is to show that the decay in probability with increasing cost outweighs the number of possible bridge systems of that cost. \cref{lem:finalcontourbound} expresses the decay in probability with increasing cost, and the following lemma will bound the number of possible bridge systems with at most some cost $K$.
\begin{lemma}[Number of Bridge Systems]
\label{lem:numbridgesystems}
Suppose that $K$ is a nonnegative integer. We denote by $\bdgcontd(\CFGPf{L}{\nfixedvec})$ the set of all possible conotur sets $\conbad$ where $\conbad = \bdgcontd(\config)$ for some $\sigma \in \CFGPf{L}{\nfixedvec}$. We then have the following upper bound for the number of possible such contours $\conbad$ with at most some given cost $K$:
\begin{align*}
|\{\conbad \in \bdgcontd(\CFGPf{L}{\nfixedvec}): \Hcostbar{\conbad} \leq K\}|
\leq 72^{\left(1 + \frac{1}{2\delta}\right)L} (72^{1 + \frac{1}{2\delta}}\numcolors)^{\frac{2}{\costmatrixmin}K}.
\end{align*}
\end{lemma}

\begin{proof}
For contour set $\conbad \in \bdgcontd(\CFGPf{L}{\nfixedvec})$, there exists a configuration $\config \in \CFGPf{L}{\nfixedvec}$ such that $\conbad = \bdgcontd(\config)$. Thus there exists a set $B = \bdgvertd(\config) \subseteq V(\regionL)$ that meets the conditions given in \cref{lem:bridgesystems}.

Denote by $v_0$ the bottom-left vertex of the region $V(\regionL)$. By how $B$ is constructed, it must contain~$v_0$.
A depth-first search can be used over the connected set $B \cup \conbadbar$ to define a walk that starts at $v_0$ and visits every vertex of $B \cup \conbadbar$ (which has at most $\left(1 + \frac{1}{2\delta}\right)|\conbadbar| + L$ vertices), that has at a length at of most $2|B \cup \conbadbar| \leq 2\left(1 + \frac{1}{2\delta}\right)|\conbadbar| + 2L$.

Thus, each possible set $\conbad$ can be constructed by such a walk, followed by marking each vertex visited by the walk as either part of $\conbadbar$ or not part of $\conbadbar$ (i.e., they came from $B$). As each vertex has degree $6$, this can be done in $6^{2\left(1 + \frac{1}{2\delta}\right)|\conbadbar| + 2L} 2^{\left(1 + \frac{1}{2\delta}\right)|\conbadbar| + L} = 72^{\left(1 + \frac{1}{2\delta}\right)|\conbadbar| + L} = $ ways. The contours of $\conbad$ are then defined by taking the connected components of the vertices $\conbadbar$ in $\regionL$, and assigning a color to each of the vertices. The color assignment can be done in at most $\numcolors^{|\conbadbar|}$ ways, giving the bound in the lemma by using the fact that $|\conbadbar| \leq \frac{2}{\costmatrixmin}\Hcost{\conbad}$ from \cref{lem:weightedcostratio}.
\end{proof}
Finally, to prove \cref{thm:main}, it suffices to show that for any $\alpha > 1$ and $\delta > 0$, there exists a sufficiently large $\betaminF = \betaminF(\alpha,\delta,\rhovec,\numcolors,\costmatrixmin,\costmatrixmax)$ and $\LminF = \LminF(\beta,\alpha,\delta,\rhovec,\numcolors,\costmatrixmin,\costmatrixmax)$ such that for all $\beta \geq \betaminF$ and $L \geq \LminF$, a configuration drawn from the distribution $\dtbnf$ on $\regionL$ contains a consistent set of contours $\confinal$ such that:
\begin{enumerate}
\item $\Hcostbar{\contourset} \leq \alpha \hamil(\optconfig)$,
\item For each component $V_k$ of the subgraph induced by $V(\regionL) \setminus \contoursetbar$ with label $j_k$, there are at most $\delta|V_k|$ particles of colors other than $j_k$,
\end{enumerate}
where $\optconfig$ denotes the configuration of $\CFGPf{L}{\nfixedvec}$ that minimizes $\hamil(\config)$ on he set of configurations $\CFGPf{L}{\nfixedvec}$. %

\begin{proof}[Proof of \cref{thm:main}]
We start by observing that by \cref{lem:bridgesystems}, we can construct a $\delta$-bridge system for configuration $\config \in \CFGPf{L}{\nfixedvec}$, which produces a consistent set of contours $\bdgcontd(\config)$ that satisfies the second condition. Thus, we only need to show that with high probability, the $\delta$-bridge system we construct satisfies $\Hcostbar{\bdgcontd(\config)} \leq \alpha \hamil(\optconfig)$.

We define $\alphaA = \frac{\alpha+1}{2} > 1$ and $\alphaB = \frac{2\alpha}{\alpha+1} > 1$, so that $\alphaA \alphaB = \alpha$. We use these values of $\alphaA$ and $\alphaB$ when applying Lemmas~\ref{lem:finalcontourbound} and~\ref{lem:approximatebestcontour} respectively.
As we know that $\Hcostbar{\congood} \leq \alphaB \hamil(\optconfign{\rhovec}) = \alphaB \hamil(\optconfig)$ by \cref{lem:approximatebestcontour}, a contour set $\confinal$ that satisfies $\Hcostbar{\confinal} \leq \alphaA \Hcostbar{\congood}$ will satisfy $\Hcostbar{\confinal} \leq \alpha \hamil(\optconfig)$.

As the contour set $\bdgcontd(\config)$ for any configuration $\config$ is always closed under contractibility in $\config$, we must have
$\{\config \in \CFGPf{L}{\nfixedvec} : \bdgcontd(\config) = \confinal\} \subseteq \CFGPcf{L}{\confinal}{\nfixedvec}$ for any contour set $\confinal$. 
We can thus bound the probability that the $\delta$-bridge system of a configuration drawn from the distribution $\dtbnf$ has cost greater than $\alpha \Hcostbar{\congood}$ as follows:
\begin{align*}
\dtbnf&\left(\left\{\config : \Hcostbar{\bdgcontd(\config)} > \alphaA \Hcostbar{\congood}\right\}\right)
\leq \sum_{k=0}^\infty \sum_{\substack{\confinal \in \bdgcontd(\CFGPf{L}{\nfixedvec}) \\ \Hcostbar{\confinal} \geq \alphaA\Hcostbar{\congood} + k \\ \Hcostbar{\confinal} \leq \alphaA \Hcostbar{\congood} + k+1}}
\dtbnf(\CFGPcf{L}{\confinal}{\nfixedvec}) \\
&\leq \sum_{k=0}^\infty |\{\confinal \in \bdgcontd(\CFGPf{L}{\nfixedvec}): \Hcostbar{\confinal} \leq \ceil{\alphaA \Hcostbar{\congood} + k + 1}\}| \cdot \sup_{\confinal: \Hcostbar{\confinal} \geq \alphaA \Hcostbar{\congood}+k}\frac{\ZPcf{L}{\confinal}{\nfixedvec}}{\ZPcf{L}{\congood}{\nfixedvec}}.
\end{align*}
Applying Lemmas~\ref{lem:finalcontourbound} and~\ref{lem:numbridgesystems}, we have:
\begin{align}
\label{eqn:energybound}
\sup_{\confinal: \Hcostbar{\confinal} \geq \alphaA \Hcostbar{\congood}+k} \frac{\ZPcf{L}{\confinal}{\nfixedvec}}{\ZPcf{L}{\congood}{\nfixedvec}}
&\leq \poly(L) \cdot \exp\left\{-\frac{\beta}{4} \left(\frac{\alphaA-1}{\alphaA}\right) \left(\alphaA \Hcostbar{\congood} +k\right)\right\},
\end{align}
and
\begin{align}
\label{eqn:entropybound}
|\{\confinal \in \bdgcontd(\CFGPf{L}{\nfixedvec}): \Hcostbar{\confinal} \leq \ceil{\alphaA \Hcostbar{\congood} + k + 1}\}|
&\leq 72^{\left(1 + \frac{1}{2\delta}\right)L} (72^{1 + \frac{1}{2\delta}}\numcolors)^{\frac{2}{\costmatrixmin}\left(\alphaA \Hcostbar{\congood} + k + 2\right)}.
\end{align}

Denote $\rhomin = \min_{i\in\colorset}\rho_i$ and $\rhobetamin = \min_{i\in\colorset}\rhobeta{i}$.
From \cref{lem:goodcontourapproximation}, we have $\rhobetamin \geq \rhomin - 2e^{-\taub \minlength/4}$. Applying Appendix~\ref{apx:allconstants} with the assumption that $\beta \geq \betamin$ gives us $\rhobetamin \geq \rhomin/2$, which means that by Lemmas~\ref{lem:weightedcostratio} and~\ref{lem:mincostconfig}, we can bound $L$ by $\Hcostbar{\congood}$ in \cref{eqn:entropybound} as follows:
\begin{align*}
\label{eqn:costLbound}
\alphaA \Hcostbar{\congood}
\geq \Hcostbar{\congood}
\geq \frac{\costmatrixmin}{2} |\congoodbar|
\geq \frac{\costmatrixmin}{2} \bdconstlow L \sqrt{\rhobetamin}
\geq \frac{\costmatrixmin}{2} \bdconstlow L \sqrt{\rhomin/2}.
\end{align*}
Now suppose that $\beta$ is large enough for the following statement to hold:
\begin{align*}
\ds 72^{\left(1 + \frac{1}{2\delta}\right) \frac{2}{\costmatrixmin \bdconstlow \sqrt{\rhomin/2}}} (72^{1 + \frac{1}{2\delta}}\numcolors)^{\frac{2}{\costmatrixmin}} < e^{\frac{\beta}{8} \left(\frac{\alphaA-1}{\alphaA}\right)}.
\end{align*}
Combining \cref{eqn:energybound} and~\cref{eqn:entropybound}, and assuming that $L$ is large enough (relative to $\beta$) for the polynomial to be absorbed by the exponential term, we can conclude the following, which we know from the opening statements of this proof implies our theorem:
\begin{align*}
\dtbnf&\left(\left\{\config : \Hcostbar{\bdgcontd(\config)} > \alphaA \Hcostbar{\congood}\right\}\right)
\leq e^{-\frac{\beta}{8}(\alphaA-1)\Hcostbar{\congood}}
= e^{-\frac{1}{32\sqrt{2}} \beta (\alpha-1) \costmatrixmin \bdconstlow \sqrt{\rhomin} \cdot L}. \qedhere
\end{align*}
\end{proof}

%% file: proof_appendix.tex
\appendix
\section{Interaction matrices of examples in Figure~\ref{fig:examples2}}
\label{apx:interactionmatrices}
This appendix contains the interaction matrices we used for the $9$-color examples in Figure~\ref{fig:examples2}.
\begin{multicols}{2}
Figure~\ref{fig:examples2}(a):
\begin{align*}
A_{\mathrm{2a}} = 2\begin{pmatrix}
0 & 1 & 2 & 2 & 4 & 4 & 4 & 4 & 4 \\
1 & 0 & 2 & 2 & 4 & 4 & 4 & 4 & 4 \\
2 & 2 & 0 & 1 & 4 & 4 & 4 & 4 & 4 \\
2 & 2 & 1 & 0 & 4 & 4 & 4 & 4 & 4 \\
4 & 4 & 4 & 4 & 0 & 1 & 2 & 2 & 4 \\
4 & 4 & 4 & 4 & 1 & 0 & 2 & 2 & 4 \\
4 & 4 & 4 & 4 & 2 & 2 & 0 & 1 & 4 \\
4 & 4 & 4 & 4 & 2 & 2 & 1 & 0 & 4 \\
4 & 4 & 4 & 4 & 4 & 4 & 4 & 4 & 0
\end{pmatrix}
\end{align*}
Figure~\ref{fig:examples2}(c):
\begin{align*}
A_{\mathrm{2c}} = ~~\begin{pmatrix}
0 & 1 & 3 & 3 & 2 & 2 & 2 & 2 & 3 \\
1 & 0 & 1 & 3 & 2 & 2 & 2 & 2 & 3 \\
3 & 1 & 0 & 1 & 2 & 2 & 2 & 2 & 3 \\
3 & 3 & 1 & 0 & 2 & 2 & 2 & 2 & 3 \\
2 & 2 & 2 & 2 & 0 & 1 & 3 & 3 & 3 \\
2 & 2 & 2 & 2 & 1 & 0 & 1 & 3 & 3 \\
2 & 2 & 2 & 2 & 3 & 1 & 0 & 1 & 3 \\
2 & 2 & 2 & 2 & 3 & 3 & 1 & 0 & 3 \\
3 & 3 & 3 & 3 & 3 & 3 & 3 & 3 & 0
\end{pmatrix}
\end{align*}
Figure~\ref{fig:examples2}(b):
\begin{align*}
A_{\mathrm{2b}} = ~~\begin{pmatrix}
0 & 1 & 2 & 3 & 4 & 4 & 4 & 4 & 2 \\
1 & 0 & 1 & 2 & 4 & 4 & 4 & 4 & 2 \\
2 & 1 & 0 & 1 & 4 & 4 & 4 & 4 & 2 \\
3 & 2 & 1 & 0 & 4 & 4 & 4 & 4 & 2 \\
4 & 4 & 4 & 4 & 0 & 1 & 2 & 3 & 2 \\
4 & 4 & 4 & 4 & 1 & 0 & 1 & 2 & 2 \\
4 & 4 & 4 & 4 & 2 & 1 & 0 & 1 & 2 \\
4 & 4 & 4 & 4 & 3 & 2 & 1 & 0 & 2 \\
2 & 2 & 2 & 2 & 2 & 2 & 2 & 2 & 0
\end{pmatrix}
\end{align*}
Figure~\ref{fig:examples2}(d):
\begin{align*}
A_{\mathrm{2d}} = \frac{2}{3}\begin{pmatrix}
0 & 2 & 2 & 2 & 2 & 2 & 2 & 2 & 5 \\
2 & 0 & 2 & 4 & 4 & 4 & 4 & 2 & 5 \\
2 & 2 & 0 & 2 & 4 & 4 & 4 & 4 & 5 \\
2 & 4 & 2 & 0 & 2 & 4 & 4 & 4 & 5 \\
2 & 4 & 4 & 2 & 0 & 2 & 4 & 4 & 5 \\
2 & 4 & 4 & 4 & 2 & 0 & 2 & 4 & 5 \\
2 & 4 & 4 & 4 & 4 & 2 & 0 & 2 & 5 \\
2 & 2 & 4 & 4 & 4 & 4 & 2 & 0 & 5 \\
5 & 5 & 5 & 5 & 5 & 5 & 5 & 5 & 0
\end{pmatrix}
\end{align*}
\end{multicols}

\section{Constants for inverse temperature and domain size}
\label{apx:allconstants}
In this appendix we specify the constants that we use as lower bounds for the inverse temperature $\beta$ and the domain side length $L$ throughout the proof.
Many of the bounds on these constants can be significantly tightened, but we opt to use loose bounds to make the lemma statements more readable.

Note that some of these statements rely on the constants $\bdconstlow$, $\Hconstupp$ from \cref{lem:mincostconfig}, as well as $\contourconst$ from \cref{lem:pointcontourcount}. While these lemmas are stated and proved later in the paper, they and the lemmas they depend on (\cref{lem:weightedcostratio}) do not rely on any other result in this paper, and they are only placed later in the paper for organizational purposes. We use $\rhomin$ to denote $\min_{i \in \colorset}\rho_i$, and the definition of $\beta'$ comes from \cref{lem:convergencecriterion}.

We define $\betamin = \betamin(\rhovec,\numcolors,\costmatrixmin,\costmatrixmax)$ to be a large enough constant such that for all $\beta \geq \betamin$, \cref{lem:balance} holds, and the following inequalities are satisfied:
\begin{itemize}
\item $\beta \geq \beta'(\contourconst,1)$, $\beta \geq \beta'(\contourconst,1/2)$
\hfill (Used in Lemmas~\ref{lem:balance}, \ref{lem:truncatedderivatives} and \ref{lem:boundaryerrorbound})
\item $\frac{6e^{\taub/8}}{(e^{\taub/8}-1)^2} \leq 1$
\hfill (Used in \cref{lem:boundaryerrorbound})
\item $e^{-\taub \minlength/4} < \frac{1}{q}$
\hfill (Used in \cref{lem:rhovecvalidity})
\item $2e^{-\taub \minlength/4} \leq \frac{1}{2} \rhomin$
\hfill (Used in Lemmas~\ref{lem:rhovecvalidity}~and~\ref{lem:goodineq} and \cref{thm:main})
\item $\hoptbound + \numcolors \leq \beta \costmatrixmax$
\hfill (Used in \cref{lem:smallchanges})
\item $77 \numcolors^{3/2} \costmatrixmax \leq e^{-\taub \minlength(1/20-1/16)}$
\hfill (Used in Lemmas~\ref{lem:regionfixedmag} and~\ref{lem:goodineq})
\item $\ds 2 + \frac{\sqrt{\numcolors}}{\bdconstlow \sqrt{\rhomin/2}} \leq e^{-\taub \minlength(1/24-1/20)}$
\hfill (Used in \cref{lem:goodineq})
\item $2\cdot \hoptbound \leq e^{-\taub \minlength/8}$
\hfill (Used in \cref{lem:grandineq})
\item $\ds \sqrt{2}e^{-\taub \minlength/2} < \frac{\costmatrixmin}{18\Hconstupp} \rhomin$
\hfill (Used in \cref{lem:approximatebestcontour})
\item $\ds  2e^{-\taub \minlength/2} \leq \min\left\{\frac{1}{6}, \frac{\costmatrixmin}{48\Hconstupp}\right\} \rhomin$
\hfill (Used in \cref{lem:approximatebestcontour})
\end{itemize}
For $\alphaA > 1$, we define $\betaminA(\alphaA) = \betaminA(\alphaA, \rhovec,\numcolors,\costmatrixmin,\costmatrixmax)$ such that for all $\beta \geq \betaminA(\alphaA)$, the following statements are true:
\begin{itemize}
\item $\beta \geq \frac{4}{\costmatrixmin} \cdot 13 e^{-\taub \minlength/8} \cdot \frac{2\alphaA}{\alphaA-1}$
\hfill (Used in \cref{lem:grandineq})
\item $\alphaA \frac{\costmatrixmin}{2} \cdot \frac{\beta}{4}\left(\frac{\alphaA-1}{\alphaA}\right) \geq \beta e^{-\taub \minlength/24}.$
\hfill (Used in \cref{lem:finalcontourbound})
\end{itemize}
For $\alphaB > 1$, we define $\betaminB(\alphaB) = \betaminB(\alphaB, \rhovec,\numcolors,\costmatrixmin,\costmatrixmax)$ such that for all $\beta \geq \betaminB(\alphaB)$, the following statements are true:
\begin{itemize}
\item $20 \numcolors e^{-\taub \minlength/2} \leq (\alphaB-1)\Hconstlow \sqrt{\rhomin}.$
\hfill (Used in \cref{lem:approximatebestcontour})
\end{itemize}
In addition, we define $\Lmin(\beta) = \Lmin(\beta,\costmatrixmin,\costmatrixmax)$ to be sufficiently large, such that for all $L \geq \Lmin(\beta)$, we have:
\begin{itemize}
\item $2\numparticles e^{-\taub L/8} \leq 1$
\hfill (Used in \cref{lem:boundaryerrorbound})
\item $\ds 3\numcolors/L^2 < \rhomin$
\hfill (Used in \cref{lem:approximatebestcontour})
\item $\sqrt{2}L e^{-\taub \minlength/2} + 1 \leq 2Le^{-\taub \minlength/2}$
\hfill (Used in \cref{lem:approximatebestcontour})
\end{itemize}